\documentclass[12pt]{article}
\usepackage{amsfonts}
\usepackage{amssymb}
\usepackage{mathrsfs}
\usepackage{srcltx}
\usepackage{amsmath}
\usepackage{amsthm}
\usepackage{amstext}
\usepackage{amsopn}
\usepackage{graphicx}
\usepackage{dsfont}
\usepackage{color}
\usepackage{bm}
\newtheorem{theorem}{Theorem}[section]
\newtheorem{lemma}[theorem]{Lemma}

\theoremstyle{definition}

\usepackage{caption}
\usepackage{subcaption}

\theoremstyle{remark}
\newtheorem{remark}[theorem]{Remark}

\numberwithin{equation}{section}

\usepackage{bm}

\newcommand{\ba}{\begin{array}}
\newcommand{\ea}{\end{array}}
\newcommand{\f}{\frac}

\newcommand{\la}{\lambda}

\newcommand{\ds}{\displaystyle}

\begin{document}
\date{}
\title{ \textbf{\large{Analysis of a Reaction-Diffusion Susceptible-Infected-Susceptible Epidemic Patch Model Incorporating Movement Inside and Among Patches}\thanks{S. Chen is supported by National Natural Science Foundation of China (No. 12171117) and Shandong Provincial Natural Science Foundation of China (No. ZR2020YQ01). 
}}
}

\author{
Shanshan Chen\footnote{Email: chenss@hit.edu.cn}\\[-1mm]
{\small Department of Mathematics, Harbin Institute of Technology}\\[-2mm]
{\small Weihai, Shandong 264209, P. R. China}\\[1mm]
Yixiang Wu\footnote{Corresponding Author, Email: yixiang.wu@mtsu.edu} \\[-1mm]
{\small Department of Mathematics, Middle Tennessee State University}\\[-2mm]
{\small Murfreesboro, Tennessee 37132, USA}
}

\maketitle
\begin{abstract}
In this paper, we propose and analyze a reaction-diffusion susceptible-infected-susceptible (SIS) epidemic patch model. The individuals are assumed to reside in different patches, where they are able to move inside and among the patches. The movement of individuals inside the patches is descried by diffusion terms, and the movement pattern among patches is modeled by an essentially nonnegative matrix. We define a basic reproduction number $\mathcal{R}_0$ for the model and show that it is a threshold value for disease extinction versus persistence. The monotone dependence of $\mathcal{R}_0$ on the movement rates of infected individuals  is proved when the dispersal pattern is symmetric or non-symmetric. Numerical simulations are performed to illustrate the impact of the movement of individuals inside and among patches on the transmission of the disease.

\end{abstract}

\section{Introduction}

The outbreaks of infectious diseases have been described and studied by various differential equation models \cite{1991book,2008book,Diekmann2000}. The spatial heterogeneity of the environment \cite{hagenaars2004spatial,lloyd1996spatial} and the movement of the populations  \cite{stoddard2009role,tatem2006global} play essential roles in the transmission of infectious diseases.
 The impact of those factors on the persistence and control of diseases has been investigated by ordinary differential equation patch models \cite{Allenpatch,Arino2003,lloyd1996spatial,WangZhaoP1} and by reaction-diffusion equation models \cite{Allen,Fitz2008,WangZhao}.

Allen \emph{et al}. \cite{Allen} have proposed the following reaction-diffusion SIS epidemic model with standard incidence mechanism:
\begin{equation}\label{patc2}
\begin{cases}
\ds\f{\partial  S}{\partial t}=d_S\Delta S-\ds\f{\beta(x) S I}{ S+ I}+\gamma(x)  I, &x\in\Omega, t>0,\smallskip\\
\ds\f{\partial   I}{\partial t}=d_I\Delta I+\ds\f{\beta(x) S I}{  S+ I}-\gamma(x)  I,&x\in\Omega, t>0,\smallskip\\
\ds\f{\partial S}{\partial \nu}=\ds\f{\partial I}{\partial \nu}=0, &x\in\partial\Omega,t>0,
\end{cases}
\end{equation}
where $S(x, t)$ and $I(x, t)$ are the density of susceptible and infected individuals at position $x\in\Omega$ and time $t$, respectively; $\beta$ and $\gamma$ are the disease transmission and recovery rates; respectively; $d_S$ and $d_I$ are the movement rates of susceptible and recovered individuals, respectively. In \cite{Allen}, the authors define a basic reproduction number $\mathcal{R}_0$ for the model, which  is decreasing in $d_I$. They show that $\mathcal{R}_0$ serves as a threshold value for disease extinction versus persistence and the disease component of the endemic equilibrium (EE) (i.e. positive equilibrium) of \eqref{patc2} approaches zero as $d_S$ approaches zero as long as $\beta-\gamma$ changes sign. Biologically, this means that one may control the disease by limiting the movement of susceptible individuals. This research has inspired a series of related works such as investigating the impact of limiting $d_I$ \cite{Peng2009} and seasonality \cite{BaiPeng2018,PengZhao,ZhangZhao2021}, replacing standard incidence mechanism by mass action mechanism \cite{castellano2022effect,Cui,DengWu, WuZou}, introducing a demographic structure into the model \cite{LiPeng} and many others  
\cite{ChenShi2020,CuiLamLou,CuiLou,Jiang,KuoPeng,Li2018,WebWu2018,PengLiu,Peng2013,SongLou,Tuncer2012}.

Our study is also inspired by the following discrete-space version of \eqref{patc2} by Allen \emph{et al}.\cite{Allenpatch}:
\begin{equation}\label{patc1}
\begin{cases}
\ds\f{d  S_j}{dt}=d_S\sum_{k\in\Omega} (L_{jk} S_k-  L_{kj} S_j)-\ds\f{\beta_j S_j I_j}{ S_j+ I_j}+\gamma_j  I_j, &j\in\Omega, t>0,\\
\ds\f{d  I_j}{dt}=d_I\sum_{k\in\Omega}(L_{jk} I_k-  L_{kj} I_j)+\ds\f{\beta_j S_j I_j}{  S_j+ I_j}-\gamma_j  I_j,&j\in\Omega, t>0,
\end{cases}
\end{equation}
where $\Omega=\{1,2,\dots,n\}$  is the collection of patches that the population live in; $L_{jk}\ge0 $ is the degree of  movement  from patch $k$ to patch $j$ for $j,k\in\Omega$ and $k\neq j$. If the \emph{connection matrix} $L=(L_{jk})$ is symmetric,  the asymptotic profile of the endemic equilibria of \eqref{patc1} as $d_S\rightarrow 0$ is studied in \cite{Allenpatch}, and the case $d_I\rightarrow 0$ is considered in \cite{LiPeng}.  The monotonicity of $\mathcal{R}_0$ in $d_I$ is proposed as an open problem in \cite{Allen}, which has been verified later in \cite{chen2020asymptotic,gao2019travel,gao2020fast}. If matrix $L$ is not symmetric, many parallel results in \cite{Allenpatch, LiPeng} are proved to hold in \cite{chen2020asymptotic}. 
For more references on patch epidemic models, we refer the interested readers to \cite{AM2019,Eisenberg2013,GaoRuan,gao2019habitat,JinWang,LiShuai,LiShuai2,Driess,Tien2015,WangZhaoP1,WangZhaoP2}.

In this paper, we model the spread of an infectious disease for a population that is living in different patches (e.g. regions, communities, countries), where they are able to move among and inside the patches. The model will essentially be a combination of \eqref{patc2} and \eqref{patc1}, where the movement inside patches will be expressed by diffusion terms while the movement pattern among patches will be described by a matrix. Our model is partially motivated by  the current coronavirus disease  pandemic, where different regions or countries have different control strategies on limiting the movement of people. The simulations based on our model suggest that limiting the local movement (i.e. movement inside patches) of susceptible individuals may not eliminate the disease when there are global movement (i.e. movement among patches) of people.

Our paper is organized as follows. In section 2, we propose the model, list the assumptions and present some results about the existence and boundedness of the solutions. We define the basic reproduction number $\mathcal{R}_0$ in section 3 and show that it is a threshold value in section 4. A detailed analysis of $\mathcal{R}_0$ including its monotone
dependence on the movement rates and its limits as the movement rates approaching zero or infinity are investigated in section 5. In particular, if the connection matrix is not symmetric, we use a graph-theoretic technique to prove the monotonicity of $\mathcal{R}_0$.
In section 6, we perform some numerical simulations to illustrate the impact of the movement of susceptible and infected individuals inside and among patches on the transmission of the disease. Some results used in the proofs of the main theorems are postponed to the appendix.

\section{The model}
 Let $\Omega_j$, $j=1, \dots, n$, be the patches that the population is living in. For convenience,  let $\Omega=\{1, 2, ..., n\}$ be the collection of the labels of all the patches, where $n\ge 2$. We suppose that $\Omega_j$ is a bounded domain with smooth boundary for each $j\in\Omega$. 

Let $S_j(x_j, t)$ and $I_j(x_j, t)$ be the population density of susceptible and infected people at position $x_j\in\Omega_j$ and time $t$ for each $j\in\Omega$, respectively. As in \cite{Allen}, we  use diffusion to model the movement of populations inside each patch and standard incidence mechanism to describe the interaction between susceptible and infected populations. Moreover, as we are focusing on the impact of movement on the transmission of the disease, we do not model the demographic structure of the population and the disease induced mortality for simplicity. 
Then
considering the movement inside and among patches, we propose the following model:
\begin{equation}\label{ori}
\begin{cases}
\ds\f{\partial  S_j}{\partial t}=
d_S\sum_{k\ne j}\left(\int_{\Omega_k}D_{jk}(x_j,x_k)S_{k}(x_k,t)dx_k-\int_{\Omega_k}D_{kj}(x_k,x_j)S_j(x_j,t)dx_k\right)\\
~~~~~~~~~+d_{S_j}\Delta S_j-\ds\f{\beta_j S_j I_j}{ S_j+ I_j}+\gamma_j  I_j,\;\;\;x_j\in\Omega_j, t>0, j\in\Omega,\\
\ds\f{\partial   I_j}{\partial t}=
d_I\sum_{k\ne j}\left(\int_{\Omega_k}D_{jk}(x_j,x_k)L_{k}(x_k,t)dx_k-\int_{\Omega_k}D_{kj}(x_k,x_j)I_j(x_j,t)dx_k\right)\\
~~~~~~~~~+d_{I_j}\Delta I_j+\ds\f{\beta_j S_j I_j}{  S_j+ I_j}-\gamma_j  I_j,\;\;\;x_j\in\Omega_j, t>0, j\in\Omega,  
\end{cases}
\end{equation}
where  $d_{S_j}, d_{I_j}>0$ are the dispersal rate of susceptible and infectious individuals inside patch $j$, respectively; $\beta_j$ and $\gamma_j$ are disease transmission and recovery rates, respectively; $d_S, d_I>0$ are  the uniform dispersal rate of  susceptible and infectious individuals among patches, respectively;
and
$D_{jk}(x_j,x_k): \Omega_j\times \Omega_k\to \mathbb{R}^+$ denotes the degree of the population movement
from position $x_k\in \Omega_k$  to position $x_j\in \Omega_j$. 

For simplicity,  we suppose that $D_{jk}(x_j,x_k)$ is spatially homogeneous and  $D_{jk}=L_{jk}/|\Omega_j|$, where $|\Omega_j|$ is the volume of $\Omega_j$  for each $j\in\Omega$. Then $L_{jk}$ is the degree of movement from point $x_k\in\Omega_k$ to patch $j$. Moreover,  $\int_{\Omega_k}D_{jk}(x_j,x_k)S_{k}(x_k,t)dx_k=L_{jk}\bar S_k/|\Omega_j|$ describes the movement of susceptible individuals from patch $k$ to  point $x_j\in\Omega_j$; $\int_{\Omega_k}D_{kj}(x_k,x_j)S_j(x_j,t)dx_k=L_{kj}S_j(x_j, t)$ is the movement of susceptible individuals from point $x_j\in\Omega_j$ to patch $k$. Here,  $\bar S_j$ and $\bar I_j$ are
the total susceptible and infected populations at patch $j$, respectively, i.e., 
$$
\bar S_j(t): = \int_{\Omega_j} S_j(x, t) \ \ \text{ and }\ \  \bar I_j(t): = \int_{\Omega_j} I_j(x, t).
$$
(In the following, $\bar\phi$ is understood as $\bar\phi:=\int_{\Omega_j}\phi$ for any $\phi:\Omega_j\to\mathbb{R}$.)  Let $L_{jj}:=-\sum_{k\neq j} L_{kj}$, which is the total degree of movement  out from point $j$.
Then, the matrix $L=(L_{jk})$ describes the movement pattern of individuals among patches. The sum of each column of $L$ is zero, which means that there is no population loss during the movement among patches. The aforementioned simplifications of \eqref{ori} lead to the following reaction-diffusion patch model:
\begin{equation}\label{patc3}
\begin{cases}
\ds\f{\partial  S_j}{\partial t}=d_{S_j}\Delta S_j+ d_S\sum_{k\neq j} \left(\f{ L_{jk}\bar S_k}{|\Omega_j|}-  L_{kj}  S_j\right)-\ds\f{\beta_j S_j I_j}{ S_j+ I_j}+\gamma_j  I_j, &x\in\Omega_j, t>0, j\in\Omega,\\
\ds\f{\partial   I_j}{\partial t}=d_{I_j}\Delta I_j+d_I\sum_{k\neq j} \left(\f{ L_{jk}\bar I_k}{|\Omega_j|}-  L_{kj}  I_j\right)+\ds\f{\beta_j S_j I_j}{  S_j+ I_j}-\gamma_j  I_j,&x\in\Omega_j, t>0, j\in\Omega,
\end{cases}
\end{equation}
Here and after, we abuse the notation by using $x$ instead of $x_j$.   We impose homogeneous Neumann boundary condition on $\partial\Omega_j$:
\begin{equation}\label{pat3B}
    \ds\f{\partial S_j}{\partial \nu}=\ds\f{\partial I_j}{\partial \nu}=0, \ \ \ x\in\partial\Omega_j, t>0, j\in\Omega,
\end{equation}
where $\nu$ is the outward normal to $\Omega_j$. 
Summing up the equations of $S_j$ and $I_j$ in \eqref{patc3} and integrating on $\Omega_j$, we obtain
$$
\f{d}{dt}\sum_{j\in\Omega}(\bar S_j(t)+\bar I_j(t))=0, \ t>0.
$$
This means that the total population remains a constant, i.e.,
\begin{equation}\label{total}
\sum_{j\in\Omega}(\bar S_j(t)+\bar I_j(t))=\sum_{j\in\Omega}(\bar S_j(0)+\bar I_j(0)), \ \ \text{for all } t>0.
\end{equation}

We impose the following assumptions:
\begin{itemize}
\item[(A0)] For each $j\in\Omega$, $S_j(x, 0)$ and $I_j(x, 0)$ are nonnegative continuous functions on $\bar\Omega_j$ with $\sum_{j\in\Omega}(\bar S_j(0)+\bar I_j(0)):=N>0$.
\item[(A1)] For each $j\in\Omega$, $\beta_j$ and $\gamma_j$ are nonnegative H\"older continuous functions on $\Omega_j$  with $\bm\beta=(\beta_1, \dots, \beta_n)\neq\bm 0$ and $\bm\gamma=(\gamma_1, \dots, \gamma_n)\neq\bm 0$;

\item[(A2)] The matrix $L:=(L_{jk})_{n\times n}$ is essentially nonnegative and irreducible, and the sum of each column of $L$ is zero. 
\end{itemize} 

\begin{remark}
A real square matrix is called \emph{essentially nonnegative} if all the off-diagonal entries are non-negative. A matrix $E$ is reducible if there is a permutation matrix $U$ such that
$$
UEU^{-1}=
\begin{pmatrix}
E_1 & 0\\
* & E_2
\end{pmatrix}
$$
where $E_1$, $E_2$ are square submatrices. If $E$ is not reducible, it is \emph{irreducible}. 
\end{remark}

We have the following result about the existence,  uniqueness, boundedness of the solutions of the model:
\begin{theorem}\label{theorem_bound}
Suppose that $(A0)-(A2)$ hold. Then  model \eqref{patc3}-\eqref{total} has a unique nonnegative global classical solution. Moreover, there exists $M_1>0$ depending on initial data such that 
\begin{equation}\label{b1}
\sum_{j\in\Omega}\left(\|S_j(\cdot, t)\|_{L^\infty(\Omega_j)}+\|I_j(\cdot, t)\|_{L^\infty(\Omega_j)}\right)<M_1,
\end{equation}
and  there exists $M_2>0$ independent on initial data such that 
\begin{equation}\label{b2}
\limsup_{t\to\infty}\sum_{j\in\Omega}\left(\|S_j(\cdot, t)\|_{L^\infty(\Omega_j)}+\|I_j(\cdot, t)\|_{L^\infty(\Omega_j)}\right)<M_2.
\end{equation}
\end{theorem}
\begin{proof}
The local existence and uniqueness of nonnegative solution follows from a standard fixed point theory argument. By \eqref{total}, the solution is bounded in $L^1$ norm. Then by  \cite[Lemma 3.1]{PengZhao} (also see \cite{dung1997dissipativity}), the solution exists globally and there exist $M_1, M_2>0$ such that \eqref{b1}-\eqref{b2} holds. 
\end{proof}

\section{Disease free equilibrium and definition of $\mathcal{R}_0$}

In this section, we study the basic reproduction number of the model. As usual, we first consider the \emph{disease free equilibrium} (DFE), i.e. the equilibrium with  $\bm I$ component equalling zero. The equilibrium of \eqref{patc3}-\eqref{total} is a solution of the following nonlocal elliptic problem: 
\begin{equation}\label{eql}
\begin{cases}
\ds 0=d_{S_j}\Delta S_j+ d_S\sum_{k\neq j} \left(\f{ L_{jk}\bar S_k}{|\Omega_j|} -L_{kj}S_j\right)-\ds\f{\beta_j S_j I_j}{ S_j+ I_j}+\gamma_j  I_j, &x\in\Omega_j,  j\in\Omega,\\
\ds 0=d_{I_j}\Delta I_j+d_I\sum_{k\neq j}\left(\f{ L_{jk}\bar I_k}{|\Omega_j|} -L_{kj}I_j\right)+\ds\f{\beta_j S_j I_j}{  S_j+ I_j}-\gamma_j  I_j,&x\in\Omega_j,  j\in\Omega,\\
\ds\f{\partial S_j}{\partial \nu}=\ds\f{\partial I_j}{\partial \nu}=0, &x\in\partial\Omega_j,  j\in\Omega,\\
\sum_{j\in\Omega_j} (\bar S_j+\bar I_j)=N. 
\end{cases}
\end{equation}



For a bounded linear operator $E$ on a Banach space, we let $s(E)$ be the spectral bound of $E$ and $r(E)$ be the spectral radius of $E$. 
By the Perron-Frobenius theorem, we have the following result on the  matrix $L$:
\begin{lemma}\label{simp}
 Suppose that $(A2)$ holds. Then $s(L)=0$ is a simple
eigenvalue of $L$ corresponding to a  positive eigenvector $\bm\alpha=(\alpha_1,  ..., \alpha_n)^T$ with $\sum_{j=1}^n \alpha_j=1$. Moreover, $L$ has no other eigenvalue corresponding with a nonnegative eigenvector.
\end{lemma}

Throughout the paper,  let $\bm\alpha$ be the positive eigenvector of $L$ speciefied in Lemma \ref{simp}. Then we can find the DFE of the model:

\begin{lemma}\label{lemma_DFE}
 Suppose that $(A0)-(A2)$ hold. Then \eqref{patc3}-\eqref{total} has a unique DFE $E_0:=({\bm S},  {\bm I})$, where
$$
 {\bm S}=\left(\frac{\alpha_1 N}{|\Omega_1|},  ..., \ \frac{\alpha_n N}{|\Omega_n|}\right) \ \ \text{and}\ \ \  {\bm I}=(0, ..., 0).
$$
\end{lemma}
\begin{proof}
By the definition of DFE, ${\bm I}=\bm 0$ and ${\bm S}$ satisfies 
\begin{equation}\label{DFE}
d_{S_j}\Delta S_j+ d_S\sum_{k\neq j} \left(\f{ L_{jk}\bar S_k}{|\Omega_j|} -L_{kj}S_j \right)=0, \ \ j\in\Omega.
\end{equation}
Integrating \eqref{DFE} on $\Omega_j$, we obtain 
$$
\sum_{k\in\Omega} L_{jk}\bar S_k=0, \ \ j\in\Omega.
$$
Then by  $\sum_{j\in\Omega} \bar S_j=N$ and Lemma \ref{simp}, we have  $\bar S_j=N\alpha_j$ for $j\in \Omega$.  It follows from \eqref{DFE} that
\begin{equation}\label{DFE1}
   d_{S_j}\Delta S_j- d_S\sum_{k\neq j} L_{kj}S_j= -d_SN\sum_{k\neq j} \f{ L_{jk}}{|\Omega_j|} \alpha_k, \ \ j\in\Omega. 
\end{equation}
By the assumptions on $L$,  we know $\sum_{k\neq j} L_{kj}>0$ for each $j\in\Omega$. 
It is easy to see that $S_j=\alpha_j N/|\Omega_j|$, $j\in\Omega$, is the unique solution of \eqref{DFE1} satisfying the homogeneous Neumann boundary condition. 
\end{proof}

We follow \cite{thieme2009spectral,WangZhao} to define the basic reproduction number of the model. Linearizing the model at $E_0$, we obtain the following system:
\begin{equation}\label{linear}
\begin{cases}
\ds\f{\partial  I_j}{\partial t}=d_{I_j}\Delta I_j+d_I\sum_{k\neq j} \left(\f{ L_{jk}\bar I_k}{|\Omega_j|}- L_{kj}  I_j\right)+\ds{\beta_j  I_j}-\gamma_j  I_j,&x\in\Omega_j, t>0, j\in\Omega,\\
\ds\f{\partial I_j}{\partial \nu}=0, &x\in\partial\Omega_j, t>0, j\in\Omega.
\end{cases}
\end{equation}
 Let $\{T(t)\}$ be the semigroup on $X:=\Pi_{j=1}^nC(\bar\Omega_j)$ induced by the solution of
\begin{equation}\label{evo}
\begin{cases}
\ds\f{\partial   I_j}{\partial t}=d_{I_j}\Delta I_j+d_I\sum_{k\neq j} \left(\f{L_{jk}\bar I_k}{|\Omega_j|}-  L_{kj}  I_j\right)-{\gamma_j  I_j},&x\in\Omega_j, t>0, j\in\Omega,\\
\ds\f{\partial I_j}{\partial \nu}=0, &x\in\partial\Omega_j, t>0, j\in\Omega.
\end{cases}
\end{equation}
Let $\mathcal{L}: X\rightarrow X$ be a linear operator defined by
$$
\mathcal{L}(\bm\phi):= \int_{0}^\infty T(t)(\beta_1\phi_1, ..., \beta_n\phi_n)dt, \ \ \bm\phi=(\phi_1, ..., \phi_n)\in X.
$$
Then the basic reproduction number $\mathcal{R}_0$ is defined as the spectral radius of $\mathcal{L}$, i.e.,
\begin{equation}\label{Rdef}
\mathcal{R}_0=r(\mathcal{L}).
\end{equation}

Then we have the following characterizations of $\mathcal{R}_0$:
\begin{lemma}\label{lemma_pe}
Suppose that $(A1)-(A2)$ hold. Let $\mathcal{R}_0$ be defined by \eqref{Rdef}. Then $\mathcal{R}_0>0$ and $\mathcal{R}_0-1$ has the same sign as $\lambda_1$, where $\lambda_1$ is the principal eigenvalue (i.e., an eigenvalue that corresponds with a positive eigenvector) of the following problem
\begin{equation}\label{eig}
\begin{cases}
\lambda   \varphi_j=d_{I_j}\Delta \varphi_j+d_I\sum_{k\neq j} \left(\ds\f{L_{jk}\bar \varphi_k}{|\Omega_j|}-  L_{kj}  \varphi_j\right)+(\beta_j-\gamma_j)  \varphi_j,&x\in\Omega_j,  j\in\Omega,\\
\ds\f{\partial \varphi_j}{\partial \nu}=0, &x\in\partial\Omega_j,  j\in\Omega,
\end{cases}
\end{equation}
and $\lambda_1$ is the only eigenvalue corresponding with a positive eigenvector, which is unique up to multiplying by a constant. Moreover, $\tilde\lambda_1:=1/\mathcal{R}_0$ is a principal eigenvalue  of the following problem
\begin{equation}\label{eig1}
\begin{cases}
-\lambda  \beta_j \varphi_j=d_{I_j}\Delta \varphi_j+d_I\sum_{k\neq j} \left(\ds\f{L_{jk}\bar \varphi_k}{|\Omega_j|}-  L_{kj}  \varphi_j\right)-\gamma_j  \varphi_j,&x\in\Omega_j,  j\in\Omega,\\
\ds\f{\partial \varphi_j}{\partial \nu}=0, &x\in\partial\Omega_j,  j\in\Omega.
\end{cases}
\end{equation}
\end{lemma}
\begin{proof}
Let $A, B: X\to X$ be given by 
$$
(A(\bm u))_j=d_{I_j}\Delta u_j+d_I\sum_{k\neq j} \left(\f{L_{jk}\bar u_k}{|\Omega_j|}-  L_{kj}  u_j\right)+{(\beta_j-\gamma_j)  u_j}, \ \ j\in\Omega,
$$
and 
$$
(B(\bm u))_j=d_{I_j}\Delta u_j+d_I\sum_{k\neq j} \left(\f{L_{jk}\bar u_k}{|\Omega_j|}-  L_{kj}  u_j\right)-\gamma_j  u_j, \ \ j\in\Omega,
$$
where $\bm u=(u_1, \dots, u_n)\in D(A)=D(B)\subset X$. Let $C: X\to X$ be given by 
$$
C(\bm u)=(\beta_1 u_1, \dots, \beta_n u_n).
$$
Then $A=B+C$  is a positive perturbation of $B$. By Lemma \ref{lemma_positive},  $A$ and $B$ are resolvent-positive operators (see the definition in the appendix). 
By the assumptions on $\gamma_j$ and Lemma \ref{lemma_B}, we have $s(B)<0$. Therefore,
$$
\int_0^\infty T(t)dt=-B^{-1}\ \ \text{and} \ \ \mathcal{L}=-B^{-1}C.
$$
By Lemma \ref{lemma_thieme}, $s(A)$ has the same sign as $r(\mathcal{L})-1=\mathcal{R}_0-1$. By Lemma \ref{lemma_positive}, $\lambda_1=s(A)$ is the principal eigenvalue of \eqref{eig} corresponding with a positive eigenvector, which is unique up to multiplying by a constant. 

By Lemma \ref{lemma_positive}, $-B^{-1}$ is compact and   strongly positive. Let $\mathds{1}=(1, \cdots, 1)$. Since $C\mathds{1}$ is nonnegative and nontrivial, there exists $\alpha>0$ such that $-B^{-1}C\mathds{1}\ge \alpha\mathds{1}$. Therefore, we have $\mathcal{R}_0=r(\mathcal{L})=r(-B^{-1}C)\ge \alpha>0$. Moreover since $-B^{-1}C$ is compact, by the Krein-Rutman theorem \cite[Theorem 19.2]{deimling2010nonlinear}, $\mathcal{R}_0=r(\mathcal{L})$ is a principal eigenvalue of $\mathcal{L}$, i.e., an eigenvalue that corresponds with a positive eigenvector. Hence, $\tilde \lambda_1=1/\mathcal{R}_0$ is a principal eigenvalue of the problem $B\bm\varphi=-\lambda C\bm\varphi$, which is  
\eqref{eig1}.
\end{proof}

\begin{remark}
If $\beta_j(x)>0$ for all $x\in\Omega_j$ and $j\in\Omega$, then  $-B^{-1}C$ is strongly positive. Then $\tilde\lambda_1:=1/\mathcal{R}_0$ is the unique eigenvalue of \eqref{eig1} that corresponds with a positive eigenvector. 
\end{remark}

\section{Threshold dynamics}
In this section, we show that $\mathcal{R}_0$ serves as a threshold value for the global dynamics of model \eqref{patc3}-\eqref{total}.  
Firstly, we show that $E_0$ is globally attractive if $\mathcal{R}_0<1$.
\begin{theorem}\label{theorem_R1}
Suppose that $(A0)-(A2)$ hold.
If $\mathcal{R}_0<1$, then the solution of \eqref{patc3}-\eqref{total} satisfies 
$$
\lim_{t\to\infty}  \left\|S_j(\cdot, t)-{\alpha_j N}/{|\Omega_j|}\right\|_{L^\infty(\Omega_j)}=\lim_{t\to\infty} \|I_j(\cdot, t)\|_{L^\infty(\Omega_j)}=0, \ \ \text{for all } j\in\Omega.
$$
\end{theorem}
\begin{proof}
By \eqref{patc3}, we have 
\begin{equation*}
\ds\f{\partial  I_j}{\partial t}\le d_{I_j}\Delta I_j+d_I\sum_{k\in\Omega}\left(\frac{L_{jk} \bar I_k}{|\Omega_j|}-  L_{kj} I_j\right)+ (\beta_j-\gamma_j)  I_j,  \ \ j\in\Omega.
\end{equation*}
Let $(\tilde I_1, \dots, \tilde I_n)$ be the solution of 
\begin{equation}\label{pc}
\begin{cases}
\ds\f{\partial \tilde I_j}{\partial t}=d_{I_j}\Delta \tilde I_j+d_I\sum_{k\in\Omega}\left(\frac{L_{jk}\bar{{\tilde  {I_k}}}}{|\Omega_j|}-  L_{kj}\tilde  I_j\right)+(\beta_j-\gamma_j)  \tilde  I_j,&x\in\Omega_j, t>0, \ j\in\Omega, \\
\ds \frac{\partial \tilde I_j}{\partial \nu}= 0 &x\in\partial\Omega_j, t>0,\ j\in\Omega,\\
\ds \tilde I_j(x, 0)= I_j(x, 0) &x\in\Omega_j, \ j\in\Omega.
\end{cases}
\end{equation}
Since  $\mathcal{R}_0<1$, by the the proof of Lemma \ref{lemma_pe}, we  have $\lambda_1=s(A)<0$, where $A$ is defined as in the proof of Lemma \ref{lemma_pe}. By \cite[Propositions 4.12-4.13]{webb1985theory} or \cite[Corrolary 4.3.12]{engel1999one}, there exists $C>0$ such that 
$$ 
 \tilde I_j(x, t)\le Ce^{\lambda_1t/2}, \ \ \text{for all } x\in\bar\Omega_j, t>0\ \text{and}\   j\in\Omega.
$$
By Lemma \ref{lemma_comp}, we have $I_j(x, t)\le \tilde I_j(x, t)$ for all  $x\in\bar\Omega_j$, $t>0$ and  $j\in\Omega$. Therefore,
 $I_j(x, t)\to 0$ uniformly in $x\in\bar\Omega_j$ as $t\to\infty$. Then the equations of $S_j$ can be written as 
\begin{equation}\label{pc1}
\begin{cases}
\ds\f{\partial  S_j}{\partial t}=d_{S_j}\Delta  S_j+d_S\sum_{k\in\Omega}\left(\frac{L_{jk}  \bar S_k}{|\Omega_j|}-  L_{kj}  S_j\right)+f_j(x, t),&x\in\Omega_j, t>0, \ j\in\Omega, \\
\ds \frac{\partial S_j}{\partial \nu}= 0 &x\in\partial\Omega_j, t>0, \ j\in\Omega,
\end{cases}
\end{equation}
 where $f_j(x, t)$ satisfies $|f_j(x, t)|\le Ce^{\lambda_1 t/2}$ for all $x\in\bar\Omega_j$ for some $C>0$. Integrating the first equation in \eqref{pc1} over $\Omega_j$, we obtain 
\begin{equation}\label{pc2}
\ds\f{d  \bar S_j}{d t}=d_I\sum_{k\in\Omega} L_{jk}  \bar S_k+\bar f_j(t), \ \ j\in\Omega. 
\end{equation}
By $\sum_{j\in\Omega} (\bar S_j+\bar I_j)=N$ and $I_j(\cdot, t)\to 0$, we have $\sum_{j\in\Omega} \bar S_j\to N$. By $\bar f_j(t)\le Ce^{\lambda_1t/2}$ and \eqref{pc2}, we obtain $\bar S_j(t)\to \alpha_j N$ as $t\to\infty$ for all $j\in\Omega$. Therefore, the following is a limiting system of \eqref{pc1}:
\begin{equation}\label{pcl}
\begin{cases}
\ds\f{\partial  \tilde S_j}{\partial t}=d_{S_j}\Delta  \tilde S_j+d_I\sum_{k\in\Omega}\left(\frac{ L_{jk}  \alpha_k N}{|\Omega_j|}-  L_{kj}  \tilde S_j\right), &x\in\Omega_j, t>0,\ j\in\Omega, \\
\ds \frac{\partial \tilde S_j}{\partial \nu}= 0 &x\in\partial\Omega_j, t>0,\ j\in\Omega.
\end{cases}
\end{equation}
It is easy to see that $\tilde S_j(x, t)\to \alpha_jN/|\Omega|$ uniformly for $x\in\Omega_j$ as $t\to\infty$, where $\tilde S_j$ is any solution of \eqref{pcl}. By the theory of asymptotically autonomous semiflows (\cite{thieme1992convergence}), we have 
$S_j(x, t)\to \alpha_jN/|\Omega|$ uniformly for $x\in\Omega_j$ as $t\to\infty$. 
\end{proof}

Then we show that the solutions are uniformly persistent if $\mathcal{R}_0>1$.

\begin{theorem}
Suppose that $(A0)-(A2)$ hold. If $\mathcal{R}_0>1$, then there exists $\varepsilon_0>0$ such that  the solution of \eqref{patc3}-\eqref{total} satisfies 
\begin{equation}\label{limits}
\liminf_{t\to\infty}  \min_{x\in\bar\Omega_j} S_j(x, t)>\varepsilon_0 \ \  \text{and}\ \  \liminf_{t\to\infty}  \min_{x\in\bar\Omega_j}I_j(x, t)>\varepsilon_0, \ \ \text{for all } j\in\Omega,
\end{equation}
for any initial data $(\bm S_0, \bm I_0)$ with $\bm I_0\not\equiv 0$. 
Moreover, the model has a positive equilibrium.
\end{theorem}
\begin{proof}  
 Let 
\begin{equation*}
\textstyle Y=\left\{(\bm S, \bm I)\in \left[\prod_{j\in\Omega} C(\Omega_j, \mathbb{R}+) \right]^2: \ \ \sum_{j\in\Omega}(\bar S_j+\bar I_j)=N\right\},
\end{equation*}
and
$$
\partial Y_0=\{(\bm S, \bm I)\in Y: \ \ \bm I\equiv 0\}.
$$
Denote $Y_0=Y\backslash \partial Y_0$, which is relatively open in $Y$.
 Let $\Phi(t): Y\rightarrow Y$ be the semiflow induced by the solution of  model \eqref{patc3}-\eqref{total}, i.e. $\Phi(t)(\bm S_0, \bm I_0)=(\bm S(\cdot, t), \bm I(\cdot, t))$ for all $t\ge 0$, where $(\bm S, \bm I)=((S_1, \dots,  S_n), (I_1, \dots,  I_n))$ is the solution of \eqref{patc3}-\eqref{total} with initial data $(\bm S_0, \bm I_0)\in Y$.

\underline{Claim 1}.  $Y_0$ is positively invariant with respect to $\Phi(t)$, i.e. $\Phi(t)Y_0\subseteq Y_0$ for all $t\ge 0$.

Let $(\bm S, \bm I)$ be a solution of \eqref{patc3}-\eqref{total} with $\bm I(x, 0)\not\equiv 0$.  By the nonnegativity of the solution and \eqref{patc3}, we have
\begin{equation}\label{Icc}
\begin{cases}
\ds\f{\partial   I_j}{\partial t}\ge d_{I_j}\Delta   I_j+d_I\sum_{k\in\Omega}\left(\frac{L_{jk}  \bar I_k}{|\Omega_j|}-  L_{kj}   I_j\right)-\gamma_j I_j, &x\in\Omega_j, t>0, \ j\in\Omega, \\
\ds \frac{\partial  I_j}{\partial \nu}= 0, &x\in\partial\Omega_j, t>0, \ j\in\Omega.
\end{cases}
\end{equation}
By Lemma \ref{lemma_comp}, we have 
\begin{equation}\label{pos}
I_j(x, t)>0  \ \text{for all}\ x\in\Omega_j, \ t>0 \ \text{and}\  j\in\Omega. 
\end{equation}
This implies $\Phi(t)Y_0\subseteq Y_0$ for all $t\ge 0$.

\underline{Claim 2}.  $\partial Y_0$ is positively invariant with respect to $\Phi(t)$, and for every $(\bm S_0, \bm I_0) \in \partial Y_0$ the $\omega-$limit set $\omega (\bm S_0, \bm I_0)$  of $(\bm S_0, \bm I_0)$  is the singleton $\{E_0\}$. 

If $\bm I(\cdot, 0)\equiv \bm 0$, then $\bm I(\cdot, t)\equiv \bm 0$ for all  $t>0$. Therefore, $\partial Y_0$ is positively invariant. 
To see $\omega (\bm S_0, \bm I_0)=\{E_0\}$ if $(\bm S_0, \bm I_0)\in\partial X_0$, we only need to show that $\bm S=(S_1,\dots S_n)$  satisfies $S_j(x, t)\to \alpha_j N/|\Omega_j|$ uniformly for $x\in\bar\Omega_j$ as $t\to\infty$ for all $j\in\Omega$, where $\bm S$ is a solution of 
\begin{equation}\label{pd}
\begin{cases}
\ds\f{\partial   S_j}{\partial t}=d_{S_j}\Delta   S_j+d_I\sum_{k\in\Omega}\left(\frac{L_{jk}  \bar S_k}{|\Omega_j|}-  L_{kj}   S_j\right), &x\in\Omega_j, t>0,\ j\in\Omega, \\
\ds \frac{\partial  S_j}{\partial \nu}= 0, &x\in\partial\Omega_j, t>0,\ j\in\Omega,\\
\ds \sum_{j\in\Omega}\int_{\Omega_j} S_j(x, 0) = N. &
\end{cases}
\end{equation}
The proof is similar to that for the convergence of the solutions of \eqref{pc1}, so we omit it here.

We will follow the terminology and method in \cite[Chapter 1]{ZhaoBook} to complete the proof. Define $\rho:\,\,Y \to  [0, \infty)$ by
$$
\rho (\bm u)= \min \{I_j(x): x \in \bar{\Omega}_j \ \text{and}\ j\in\Omega\},\,\,\, \bm u=(\bm S, \bm I) \in X.
$$
By the proof of claim 1, we have $\rho(\Phi(t)\bm u_0) >0$  for all $t>0$ and $\bm u_0 \in (Y_0\cap \rho^{-1}(0)) \cup \rho^{-1} (0, \infty)$. Thus, $\rho$ is a generalized distance function for the semiflow $\Phi (t): Y \to Y$.

\underline{Claim 3}.  $W^s(E_0) \cap \rho^{-1}(0, \infty)= \emptyset$, where $W^s(E_0)$ denotes the stable manifold of $E_0$.

Let $(\bm S_0, \bm I_0) \in \rho^{-1}(0,\infty)$  and  $(\bm S, \bm I)$ be the corresponding solution of \eqref{patc3}-\eqref{total} with initial data $(\bm S_0, \bm I_0)$.  Suppose to the contrary that 
$(\bm S(\cdot, t), \bm I(\cdot, t))\to E_0$ uniformly  as $ t \to \infty$. Then 
 for any $\varepsilon>0$   there exists $t_1>0$ such that $S_j(x, t)/(S_j(x, t)+I_j(x, t))>1-\varepsilon$  for all $(x, t)\in \bar\Omega_j\times [t_1, \infty)$ and $j\in\Omega$.  Hence,  $\bm I$ is  an upper solution of the following  problem 
 \begin{equation}\label{linearc}
\begin{cases}
\ds\f{\partial  \check I_j}{\partial t}=d_{I_j} \Delta \check I_j+d_I\sum_{k\neq j} \left(\f{ L_{jk}{\bar {\check I}}_k}{|\Omega_j|}- L_{kj}  \check I_j\right)+(1-\varepsilon)\beta_j  \check I_j-\gamma_j  \check I_j,&x\in\Omega_j, t>t_1, j\in\Omega,\\
\ds\f{\partial \check I_j}{\partial \nu}=0, &x\in\partial\Omega_j, t>t_1, j\in\Omega,\\
\ds \check I_j(x, t_1)=\delta \varphi_j(x, t_1), &x\in\Omega_j, j\in\Omega,
\end{cases}
\end{equation}
where $\bm \varphi=(\varphi_1, \dots, \varphi_n)$ is a positive eigenvector of 
\begin{equation}\label{eigep}
\begin{cases}
\ds\lambda   \varphi_j=d_{I_j}\Delta \varphi_j+d_I\sum_{k\neq j} \left(\f{L_{jk}\bar \varphi_k}{|\Omega_j|}-  L_{kj}  \varphi_j\right)+((1-\varepsilon)\beta_j-\gamma_j)  \varphi_j,&x\in\Omega_j,  j\in\Omega,\\
\ds\f{\partial \varphi_j}{\partial \nu}=0, &x\in\partial\Omega_j,  j\in\Omega,
\end{cases}
\end{equation}
and $\delta>0$ is small such that $\bm I(\cdot, t_1)\ge \delta \bm \varphi(\cdot, t_1)$. Since $\mathcal{R}_0>1$, by Lemma \ref{lemma_pe}, the principal eigenvalue $\lambda_1$ of \eqref{eig}  satisfies $\lambda_1>1$. So we can choose $\varepsilon>0$ small such that $\lambda_{1}^\varepsilon>1$, where $\lambda_{1}^\epsilon$ is the principal eigenvalue of \eqref{eigep}. It is easy to check that $\check I_j(x, t)=\delta e^{\lambda_1^\varepsilon t}\varphi_j(x, t)$  is the unique solution of \eqref{linearc}. By Lemma \ref{lemma_comp}, we have $I_j(x, t)\ge \check I_j(x, t)$ for all $x\in\Omega_j$ and $j\in\Omega$. This implies $I_j(x, t)\to\infty$ as $t\to\infty$, which contradicts the boundedness of $I_j$.

By Theorem \ref{theorem_bound}, $\Phi(t)$ is dissipative. It is easy to see that $\Phi(t)$ is compact. Therefore, by \cite[Theorem 1.1.3]{ZhaoBook}, $\Phi(t)$ admits a  global attractor. Then by the above claims and  the well-known abstract  persistence theory (see, e.g.,  \cite[Theorem 1.3.2]{ZhaoBook}),   $\Phi(t)$ is uniformly persistent with respect to  $(Y_0, \partial Y_0, \rho)$ in the sense that 
there exists $\varepsilon_0>0$ such that  $\lim\inf \rho(\Phi(t) \bm u_0) \ge 
\varepsilon_0$ for any $\bm u_0 \in Y$ with $\rho(\bm u_0)>0$ (i.e., $\bm u_0 \in Y_0$).  By the definition of $\rho$, we obtain \eqref{limits}.   By \cite[Theorem 1.3.7]{ZhaoBook}, $\Phi(t): Y_0\to Y_0$ admits a global attractor. Then by \cite[Theorem 1.3.11]{ZhaoBook}, $\Phi(t)$ has an equilibrium in $Y_0$. By the maximum principle for elliptic equations, it is easy to see that the equilibrium in $Y_0$ is positive. 
\end{proof}

\section{Analysis of $\mathcal{R}_0$}

In this section, we study the properties of $\mathcal{R}_0$. In the following, we set $a/0=0$ if $a=0$ and $a/0=\infty$ if $a>0$. 
First, we obtain a bound for $\mathcal{R}_0$: 
\begin{lemma}\label{lemma_Rb}
Suppose that \emph{(A1)-(A2)} hold. Then $\mathcal{R}_0$ satsifies: 
\begin{equation*}
 \inf_{x\in\bar\Omega_j, j\in\Omega} \frac{\beta_j(x)}{\gamma_j(x)}\le \mathcal{R}_0\le  \sup_{x\in\bar\Omega_j, j\in\Omega} \frac{\beta_j(x)}{\gamma_j(x)}.
\end{equation*}
\end{lemma}
\begin{proof}
Let $\varphi=(\varphi_1, ..., \varphi_n)$ be a positive eigenfunction corresponding with $\tilde\lambda_1$ for \eqref{eig1}. Integrating the first equation of \eqref{eig1} over $\Omega_j$ and summing up the equations over $j\in\Omega$, we obtain 
\begin{equation}\label{uplow0}
\tilde\lambda_1 \sum_{j\in\Omega}\int_{\Omega_j}\beta_j\varphi_j =\sum_{j\in\Omega}\int_{\Omega_j}\gamma_j\varphi_j ,
\end{equation}
and consequently, 
\begin{equation}\label{uplow}
\mathcal R_0= \ds\frac{\sum_{j\in\Omega}\int_{\Omega_j}\beta_j\varphi_j} {\sum_{j\in\Omega}\int_{\Omega_j}\gamma_j\varphi_j} ,
\end{equation}
Let 
$$
M:=\ds\sup_{x\in\bar\Omega_j, j\in\Omega} \ds\frac{\beta_j(x)}{\gamma_j(x)}.
$$ 
Clearly, if $M=\infty$ then $R_0\le M$. If $M<\infty$, replacing  $\beta_j$ by $\gamma_j\ds\frac{\beta_j}{\gamma_j}$ in \eqref{uplow}, one easily sees   $\mathcal{R}_0\le M$. Using similar arguments, we can obtain the lower bound for $\mathcal R_0$.
\end{proof}

\subsection{Symmetric $L$}
In this section, we assume that $L$ is symmetric and study the properties of $\mathcal{R}_0$. Multiplying the first equation of \eqref{eig1} by $\varphi_j$ and integrating over $\Omega_j$, we obtain 
\begin{equation}\label{phi1}
\frac{1}{\mathcal{R}_0}\int_{\Omega_j} \beta_j\varphi_j^2= d_{I_j}\int_{\Omega_j} |\triangledown\varphi_j|^2+d_I L_{jj}\left(\frac{1}{|\Omega_j|}\bar\varphi_j^2- \int_{\Omega_j} \varphi_j^2\right)-d_I\sum_{k\in\Omega} \f{L_{jk} }{|\Omega_j|}\bar \varphi_j\bar \varphi_k+ \int_{\Omega_j} \gamma_j\varphi_j^2,
\end{equation}
for each $j\in\Omega$. 
Multiplying \eqref{phi1} by $|\Omega_j|$ and summing up all the equations, we obtain
\begin{eqnarray}\label{phi2}
\frac{1}{\mathcal{R}_0}\sum_{j\in\Omega} |\Omega_j|\int_{\Omega_j} \beta_j\varphi_j^2&=&\sum_{j\in\Omega} d_{I_j}|\Omega_j|\int_{\Omega_j} |\triangledown\varphi_j|^2+d_I \sum_{j\in\Omega} L_{jj}\left(\bar\varphi_j^2- |\Omega_j|\int_{\Omega_j} \varphi_j^2\right) \nonumber\\
&&-d_I\sum_{j,k\in\Omega} {L_{jk} }\bar \varphi_j\bar \varphi_k+ \sum_{j\in\Omega}|\Omega_j|\int_{\Omega_j} \gamma_j\varphi_j^2.
\end{eqnarray}
By the H\"{o}lder's inequality, we have
\begin{equation}\label{phijsym}
\bar\varphi_j^2- |\Omega_j|\int_{\Omega_j} \varphi_j^2\le0,
\end{equation}
where the equality holds if and only if $\varphi_j$ is a constant. 
Moreover since $L$ is symmetric, it is negative semi-definite and
\begin{equation}\label{sumsym}
\sum_{j,k\in\Omega} {L_{jk} }\bar \varphi_j\bar \varphi_k\le 0,
\end{equation}
where the equality holds if and only if $\bar\varphi_1= \dots=\bar\varphi_n$.
Therefore, by \eqref{phi2}, we obtain
$$
\mathcal{R}_0=\frac{\sum_{j\in\Omega} |\Omega_j|\int_{\Omega_j} \beta_j\varphi_j^2}{f(\bm\varphi, d_I, d_{I_1}, ..., d_{I_n})+ \sum_{j\in\Omega}|\Omega_j|\int_{\Omega_j} \gamma_j\varphi_j^2},
$$
where
\begin{equation}\label{fsym}
\begin{split}
f(\bm\varphi, d_I, d_{I_1}, ..., d_{I_n}):=&\sum_{j\in\Omega} d_{I_j}|\Omega_j|\int_{\Omega_j} |\triangledown\varphi_j|^2 -d_I\sum_{j,k\in\Omega} {L_{jk} }\bar \varphi_j\bar \varphi_k\\
&+d_I \sum_{j\in\Omega} L_{jj}\left(\bar\varphi_j^2- |\Omega_j|\int_{\Omega_j} \varphi_j^2\right)\ge0.
\end{split}
\end{equation}
Here, we used  $L_{jj}<0$ for all $j\in\Omega$. 

\begin{lemma}
Suppose that \emph{(A1)-(A2)} hold and $L$ is symmetric. Let $\mathcal{R}_0$ be defined as above. Then, we have
\begin{equation}\label{var}
\mathcal{R}_0=\sup_{\substack{\bm\varphi\in H\\ 
\bm\varphi\neq \bm 0}}\left\{\frac{\sum_{j\in\Omega} |\Omega_j|\int_{\Omega_j} \beta_j\varphi_j^2}{f(\bm\varphi, d_I, d_{I_1}, ..., d_{I_n})+ \sum_{j\in\Omega}|\Omega_j|\int_{\Omega_j} \gamma_j\varphi_j^2}\right\},
\end{equation}
where $H:=\Pi_{j=1}^n H^1(\Omega_j)$.
\end{lemma}
\begin{proof}
The above computations imply 
\begin{equation*}
\mathcal{R}_0\le \sup_{\substack{\bm\varphi\in H\\ 
\bm\varphi\neq \bm 0}}\left\{\frac{\sum_{j\in\Omega} |\Omega_j|\int_{\Omega_j} \beta_j\varphi_j^2}{f(\bm\varphi, d_I, d_{I_1}, ..., d_{I_n})+ \sum_{j\in\Omega}|\Omega_j|\int_{\Omega_j} \gamma_j\varphi_j^2}\right\}.
\end{equation*}
Let the operators $B$ and $C$ be defined as in the proof of Lemma \ref{lemma_pe}. 

Firstly, we assume  $\beta_j(x)>0$ for all $x\in\Omega_j$ and $j\in\Omega$. Then it is standard to apply the spectral theory for self-adjoint compact operators to show \eqref{var} (e.g., see \cite{evans2010partial}), so we only sketch the proof here. Let $V=\Pi_{j\in\Omega}L^2(\Omega_j)$ with the inner product $<\cdot, \cdot>: V\times V\to \mathbb{R}$ defined by 
$$
<\bm u, \bm v>= \sum_{j\in\Omega} |\Omega_j| \int_{\Omega_j} \beta_j u_jv_j, \ \ \text{for\ }\bm u, \bm v\in V.
$$ 
(The positivity of $\beta_j$ guarantees that $<\cdot, \cdot>$ is an inner product). Then we can show that $B^{-1}C$ is a self-adjoint operator on $V$, i.e. $$<B^{-1}C\bm u, \bm v>=<\bm u, B^{-1}C\bm v>\;\;\text{for any}\;\; \bm u, \bm v\in V.$$ Let $\bm f=B^{-1}C\bm u$ and $\bm g=B^{-1}C\bm v$. This is equivalent to show $$<\bm f, C^{-1}B\bm g>=<C^{-1}B\bm f, \bm g>,$$ which can be verified directly using the definition of $<\cdot, \cdot>$ and the symmetry of $L$. 

Since $-B^{-1}C$ is a self-adjoint compact operator on $V$, there exists an orthonormal basis $\{\bm\phi_k\}$ of V  consisting with eigenvectors of $-B^{-1}C$ corresponding to eigenvalues $\mathcal{R}_0=\lambda_1>|\lambda_2|\ge |\lambda_3|\ge\cdots$. Moreover, all the eigenvalues are positive. To see it, we define $<\cdot, \cdot>_H: H\times H\to\mathbb{R}$ by 
\begin{equation*}
\begin{array}{lll}
<\bm\varphi,\bm\psi>_H&=&\sum_{j\in\Omega} d_{I_j}|\Omega_j|\int_{\Omega_j} \triangledown\varphi_j\cdot \triangledown\psi_j+d_I \sum_{j\in\Omega} L_{jj}\left(\bar\varphi_j \bar\psi_j- |\Omega_j|\int_{\Omega_j} \varphi_j\psi_j\right) \\
&&-d_I\sum_{j,k\in\Omega} {L_{jk} }\bar \varphi_j\bar \psi_k+ \sum_{j\in\Omega}|\Omega_j|\int_{\Omega_j} \gamma_j\varphi_j\psi_j,
\end{array}
\end{equation*}
for $\bm\varphi,\bm\psi\in H$.
Then $<\cdot, \cdot>_H$ is an inner product on $H$. We can compute 
$$
\lambda_k=<-B^{-1}C\bm\phi_k, \bm \phi_k>=<\bm\psi_k, -C^{-1}B\bm\psi_k>=<\bm\psi_k,\bm\psi_k>_H>0,
$$
where $\bm\psi_k=-B^{-1}C\bm\phi_k$. 

We claim that 
$\{\bm\phi_k\}$ is a basis of $H$. To see it,  we suppose $<\bm\phi_k, \bm w>_H=0$ for all $k\ge 1$ for some $\bm w\in H$. It suffices to show $\bm w=\bm 0$. Indeed, 
$$
0=<\bm\phi_k, \bm w>_H=<-C^{-1}B\bm\phi_k, \bm w>=\ds\frac{<\bm\phi_k, \bm w>}{\lambda_k},
$$ and so $<\bm\phi_k, \bm w>=0$ for all $k\ge 1$. Since $\{\bm\phi_k\}$ is a basis for $V$, we have $\bm w=\bm 0$. Let $\bm\varphi\in H$ with $\bm\varphi\neq\bm 0$. Then $\bm\varphi=\sum \alpha_k\bm \psi_k$ for some $\alpha_k\in\mathbb{R}$. We can compute 
\begin{eqnarray*}
\frac{<\bm w, \bm w>}{<\bm w, \bm w>_H}&=&\frac{\sum_{k, l} \alpha_k\alpha_l<\bm\psi_k, \bm\psi_l>}{\sum_{k, l}\alpha_k\alpha_l<\bm \psi_k, \bm \psi_l>_H}\\
&=& \frac{\sum \alpha_k^2}{\sum\alpha_k^2/\lambda_k}\le \lambda_1=\mathcal{R}_0.
\end{eqnarray*}
This shows 
\begin{equation}\label{var11}
\mathcal{R}_0\ge \frac{\sum_{j\in\Omega} |\Omega_j|\int_{\Omega_j} \beta_j\varphi_j^2}{f(\bm\varphi, d_I, d_{I_1}, ..., d_{I_n})+ \sum_{j\in\Omega}|\Omega_j|\int_{\Omega_j} \gamma_j\varphi_j^2}
\end{equation}
for all $\bm\varphi\in H$. So,
\begin{equation*}
\mathcal{R}_0\ge \sup_{\substack{\bm\varphi\in H\\ 
\bm\varphi\neq \bm 0}}\left\{\frac{\sum_{j\in\Omega} |\Omega_j|\int_{\Omega_j} \beta_j\varphi_j^2}{f(\bm\varphi, d_I, d_{I_1}, ..., d_{I_n})+ \sum_{j\in\Omega}|\Omega_j|\int_{\Omega_j} \gamma_j\varphi_j^2}\right\},
\end{equation*}
and \eqref{var} holds. 

Finally, we drop the positivity assumption on $\beta_j$.   Let $\beta_{j, \epsilon}=\beta_{j}+\epsilon$ for $j\in\Omega$, $C_\epsilon=(\beta_{1, \epsilon}, \dots, \beta_{n,\epsilon})$, and $\mathcal{R}_{\epsilon}=r(-B^{-1}C_\epsilon)$. Then the variational formula \eqref{var} holds for $\mathcal{R}_{\epsilon}$ with $\beta_j$ replaced by $\beta_{j, \epsilon}$.  By \cite[Section 4.3.5]{kato2013perturbation}, $\mathcal{R}_\epsilon\to \mathcal{R}_0$ as $\epsilon\to 0$.  Then by \eqref{var} for $\mathcal{R}_\epsilon$, we have
\begin{equation*}
\mathcal{R}_\epsilon\ge \frac{\sum_{j\in\Omega} |\Omega_j|\int_{\Omega_j} \beta_{j,\epsilon}\varphi_j^2}{f(\bm\varphi, d_I, d_{I_1}, ..., d_{I_n})+ \sum_{j\in\Omega}|\Omega_j|\int_{\Omega_j} \gamma_j\varphi_j^2},
\end{equation*}
for all $\bm\varphi=(\varphi_1, \dots, \varphi_n)\in H$ with $\bm\varphi\neq \bm 0$.
Taking $\epsilon\to0$, we obtain \eqref{var11}.
Therefore, \eqref{var} holds.
\end{proof}

Using the variational formula \eqref{var}, we can prove the following result.
\begin{theorem}
Suppose that \emph{(A1)-(A2)} hold and $L$ is symmetric. 
Then the following statements hold:
\begin{enumerate}
\item [\emph{(1)}] $\mathcal R_0$ is 
decreasing in  $d_I\in (0, \infty)$. Moreover, $\mathcal R_0$  is strictly decreasing in $d_I$ if and only if $(\beta_1, \dots, \beta_n)$ is not a multiple of $(\gamma_1, \dots, \gamma_n)$.
\item [\emph{(2)}]$\mathcal R_0$ is 
decreasing in $d_{I_j}\in (0, \infty)$ for each $j\in\Omega$.
\item [\emph{(2)}]$\mathcal R_0$ satisfies
\begin{equation}\label{limM}
    \lim_{\max\{d_I, d_{I_1}, ..., d_{I_n}\}\to 0} \mathcal{R}_0= \sup_{x\in\bar\Omega_j, j\in\Omega} \frac{\beta_j(x)}{\gamma_j(x)}:=M,
    \end{equation}
and
    \begin{equation}\label{limm}
    \lim_{\min\{d_I, d_{I_1}, ..., d_{I_n}\}\to\infty} \mathcal{R}_0=  \frac{\sum_{j\in\Omega}\frac{\bar\beta_j}{|\Omega_j|}}{\sum_{j\in\Omega}\frac{\bar\gamma_j}{|\Omega_j|}}:=m.
       \end{equation}
\end{enumerate}
\end{theorem}

\begin{proof}

By the variational formula \eqref{var},  $\mathcal R_0$ is decreasing in $d_I$ and $d_{I_j}$ for all $j\in\Omega$. Let $\bm \varphi$ be a  positive eigenvector of \eqref{eig1} corresponding to the principal eigenvalue
$\tilde \la_1=1/\mathcal R_0$. If $\mathcal R_0(d_I)=\mathcal R_0(\tilde d_I)$ with $\tilde d_I>d_I$, then by \eqref{fsym} and the variational formula \eqref{var} we must have that 
$$-\sum_{j,k\in\Omega} {L_{jk} }\bar \varphi_j\bar \varphi_k
+\sum_{j\in\Omega} L_{jj}\left(\bar\varphi_j^2- |\Omega_j|\int_{\Omega_j} \varphi_j^2\right)=0.$$
This combined with \eqref{phijsym} and \eqref{sumsym} implies that there exists a positive constant $C$ such that $\varphi_j=C$ for all  $x\in\bar\Omega_j$ and $j\in\Omega$. Then by \eqref{eig1} we see that $\beta_j/\gamma_j=\mathcal R_0(d_I)$ for all $j\in\Omega$, and $\mathcal R_0$ is constant. This completes the proof of (1) and (2).

Now we prove {(3)}. By Lemma \ref{lemma_Rb}, we have 
$\mathcal{R}_0\le M$.
Fix $M_1\in(0,M)$. Choose $j_0\in\Omega$, $x_0\in \Omega_{j_0}$ and $\delta>0$ such that $\beta_{j_0}(x)/\gamma_{j_0}(x)>M_1$ for any $x\in B(x_0, \delta)$. Fix $\sigma\in (0, 1)$. Let $\tilde \varphi_{j_0}\in C(\bar\Omega_{j_0}; [0, 1])$ such that 
\begin{equation*}
\tilde \varphi_{j_0}(x)=
\left\{
    \begin{array}{cc}
        0, & x\in \Omega_{j_0}\backslash B(x_0, \delta),  \\
        1, & x\in B(x_0, \sigma\delta).
    \end{array}
 \right.
\end{equation*}
Let $\tilde{\bm\varphi}=(0, ..., 0, \tilde \varphi_{j_0}, 0,...,0)$. Then by \eqref{var}, we have 
\begin{eqnarray*}
\mathcal{R}_0&\ge& \frac{\sum_{j\in\Omega} |\Omega_j|\int_{\Omega_j} \beta_j\tilde\varphi_j^2}{f(\tilde{\bm\varphi}, d_I, d_{I_1}, ..., d_{I_n})+ \sum_{j\in\Omega}|\Omega_j|\int_{\Omega_j} \gamma_{j}\tilde\varphi_{j}^2}\\
&\ge&  \frac{M_1|\Omega_{j_0}|\int_{B(x_0, \sigma\delta)} \gamma_{j_0}\tilde\varphi_{j_0}^2}{f(\tilde{\bm\varphi}, d_I, d_{I_1}, ..., d_{I_n})+ |\Omega_{j_0}|\int_{B(x_0, \delta)} \gamma_{j_0}\tilde\varphi_{j_0}^2}.
\end{eqnarray*}
This implies 
$$
\liminf_{\max\{d_I, d_{I_1}, ..., d_{I_n}\}\to 0} \mathcal{R}_0\ge  \frac{M_1\int_{B(x_0, \sigma\delta)} \gamma_{j_0}\tilde\varphi_{j_0}^2}{\int_{B(x_0, \delta)} \gamma_{j_0}\tilde\varphi_{j_0}^2}.
$$
Taking $\sigma\to 1$, we obtain 
$$
\liminf_{\max\{d_I, d_{I_1}, ..., d_{I_n}\}\to 0} \mathcal{R}_0\ge  M_1.
$$
Since $M_1\in(0,M)$ was arbitrary, we have $\liminf_{\max\{d_I, d_{I_1}, ..., d_{I_n}\}\to 0} \mathcal{R}_0\ge  M$. This combined with $\mathcal{R}_0\le M$ proves \eqref{limM}.

It reamins to prove \eqref{limm}. Taking $\varphi=(1/|\Omega_1|, ..., 1/|\Omega_n|)$ in \eqref{var}, we find  $\mathcal{R}_0\ge m$. By the monotonicity of $\mathcal{R}_0$, it suffices to consider the case $d=d_I=d_{I_1}=\cdots=d_{I_n}$ with $d\to\infty$. Let $\bm\varphi=(\varphi_1, ..., \varphi_n)$ be the positive eigenvector of \eqref{eig1} corresponding with  $\tilde\lambda_1=1/\mathcal{R}_0$ such that $\sum_{j\in \Omega}\|\varphi_j\|_{L^\infty(\Omega_j)}=1$.
By $\tilde\lambda_1\le 1/m<\infty$ and the parabolic  estimate,  $\varphi_j$ is uniformly bounded in $W^{2, p}(\Omega_j)$ for $d>1$ and $j\in\Omega$. Therefore restricted to a subsequence if necessary, we may assume $\varphi_j\to \varphi_j^*$ weakly in $W^{2, p}(\Omega_j)$ as $d\to\infty$ for some $\varphi_j^*\in W^{2, p}(\Omega_j)$. Dividing \eqref{eig1} by $d$ and taking $d\to\infty$, we find
\begin{equation}\label{eig2}
\begin{cases}
\ds\Delta \varphi_j^*+\sum_{k\neq j} \left(\f{L_{jk}\bar \varphi_k^*}{|\Omega_j|}-  L_{kj}  \varphi_j^*\right)=0, &x\in\Omega_j,  j\in\Omega,\\
\ds\f{\partial \varphi_j^*}{\partial \nu}=0, &x\in\partial\Omega_j,  j\in\Omega.
\end{cases}
\end{equation}
Similar to the proof of Lemma \ref{lemma_DFE} and using $\sum_{j\in\Omega}\|\varphi_j^*\|_{L^\infty(\Omega_j)}=1$, we can show $\varphi^*=k^*(1/{|\Omega_1|}, ..., 1/{|\Omega_n|})$, where 
$$
k^*=\frac{1}{\sum_{j\in\Omega}1/|\Omega_j|}.
$$
Then taking $d\to\infty$ in \eqref{uplow}, we obtain 
$$
\lim_{d\to\infty} \mathcal{R}_0 =\frac{\sum_{j\in\Omega}\int_{\Omega_j}\beta_j\varphi_j^* }{\sum_{j\in\Omega}\int_{\Omega_j}\gamma_j\varphi_j^* }=m.
$$
This completes the proof.
\end{proof}

\subsection{Non-symmetric $L$}
In this section, we analyze $\mathcal{R}_0$ not assuming that $L$ is symmetric. Similar to the case that $L$ is  symmetric, we will show that $\mathcal{R}_0$ is monotone decreasing in the movement rates of infected individuals $d_I, d_{I_1}, \dots, d_{I_n}$ and compute the limits of $\mathcal{R}_0$ as they approach zero or infinity.   

For convenience, we introduce a parameter $a\in \mathbb{R}$, and let $\lambda_1(d_I,d_{I_1},\dots,d_{I_n}, a)$ be the principal eigenvalue of
\begin{equation}\label{eiga}
\begin{cases}
\ds \lambda   \varphi_j=d_{I_j}\Delta \varphi_j+d_I\sum_{k\neq j} \left(\f{L_{jk}\bar \varphi_k}{|\Omega_j|}-  L_{kj}  \varphi_j\right)+(a\beta_j-\gamma_j)  \varphi_j,&x\in\Omega_j,  j\in\Omega,\\
\ds\f{\partial \varphi_j}{\partial \nu}=0, &x\in\partial\Omega_j,  j\in\Omega.
\end{cases}
\end{equation}
In the proof of the following result, we need the Tree-Cycle identity \cite[Theorem 2.2]{LiShuai2} (also see \cite{chen2021global, chen2022two} for some recent applications of it).
\begin{lemma}\label{lemma_a}
Suppose that \emph{(A1)-(A2)} hold. Let  $\la_1:=\la_1(d_I,d_{I_1},\dots,d_{I_n}, a)$ be the principal eigenvalue of \eqref{eiga}. Then the following statements hold:
\begin{enumerate}
\item [\emph{(1)}]  $\lambda_1$ is 
decreasing in  $d_I\in (0, \infty)$. Moreover, $\la_1$  is strictly decreasing in $d_I$ if and only if there does not exist a constant $C$ such that $a\beta_j(x)-\gamma_j(x)=C$ for all $x\in\Omega_j$ and $j\in\Omega$.
\item [\emph{(2)}]$\lambda_1$ is 
decreasing in $d_{I_j}\in(0, \infty)$ for each $j\in\Omega$.
\item [\emph{(3)}]$\lambda_1$ is strictly increasing in  $a\in\mathbb{R}$.
\end{enumerate}
\end{lemma}
\begin{proof}
(1) We consider the adjoint  eigenvalue problem of \eqref{eiga}:
\begin{equation}\label{eigaad}
\begin{cases}
\ds \lambda   \psi_j=d_{I_j}\Delta \psi_j+d_I\sum_{k\neq j} \left(\f{L_{kj}\bar \psi_k}{|\Omega_k|}-  L_{kj}  \psi_j\right)+(a\beta_j-\gamma_j)  \psi_j,&x\in\Omega_j,  j\in\Omega,\\
\ds\f{\partial \psi_j}{\partial \nu}=0, &x\in\partial\Omega_j,  j\in\Omega.
\end{cases}
\end{equation}
By \cite[Proposition 4.2.18]{engel1999one}, $\lambda_1$ 
is the principal eigenvalue of \eqref{eigaad}. 
Differentiating the first equation of \eqref{eigaad} with respect to $d_I$, we obtain
\begin{equation}\label{eigder}
\begin{split}
    \lambda_1'   \psi_j+\lambda_1 \psi_j'=&d_{I_j}\Delta \psi_j'+d_I\sum_{k\neq j} \left(\f{L_{kj}\bar \psi_k'}{|\Omega_k|}-  L_{kj}  \psi_j'\right)+\sum_{k\neq j} \left(\f{L_{kj}\bar \psi_k}{|\Omega_k|}-  L_{kj}  \psi_j\right)\\
    &+(a\beta_j-\gamma_j)  \psi_j',
    \end{split}
\end{equation}
where $'$ denotes $\partial/\partial d_I$. Multiplying \eqref{eigder} by $\psi_j$, multiplying \eqref{eigaad} by $\psi_j'$,  taking the difference of them and integrating it over $\Omega_j$, we obtain
\begin{equation}\label{eigenad}
\begin{split}
    \lambda_1'  \bar{\psi_j^2}=&d_I\sum_{k\neq j} \f{L_{kj}}{|\Omega_k|} \left(\bar\psi_j\bar\psi_k'-    \bar\psi_j'\bar\psi_k\right) +\sum_{k\neq j} L_{kj}\left(\ds\frac{\bar \psi_j\bar\psi_k }{{|\Omega_k|}}-\overline{\psi^2_j}\right)\\
    =&d_I\sum_{k\neq j} \f{L_{kj}\bar\psi_k\bar\psi_j}{|\Omega_k|}\left(\ds\frac{\bar\psi_k'}{\psi_k}-\ds\frac{\bar\psi_j'}{\psi_j}\right)+\sum_{k\neq j} \ds\frac{L_{kj}\bar \psi_j\bar\psi_k }{{|\Omega_k|}}\left(1-\ds\frac{|\Omega_k|\overline{\psi^2_j}}{\bar\psi_k\bar\psi_j}\right)\\
    \le& d_I\sum_{k\neq j} \f{L_{kj}\bar\psi_k\bar\psi_j}{|\Omega_k|}\left(\ds\frac{\bar\psi_k'}{\psi_k}-\ds\frac{\bar\psi_j'}{\psi_j}\right)+\sum_{k\neq j} \ds\frac{L_{kj}\bar \psi_j\bar\psi_k }{{|\Omega_k|}}\left(1-\ds\frac{w_j}{w_k}\right),\ \ j\in\Omega
    \end{split}
\end{equation}
with $w_j=(\int_{\Omega_j}\psi^2_j)^{\frac{1}{2}}/|\Omega_j|^{\frac{1}{2}}$,
where we have used the H\"{o}lder inequality in the last step and the equality holds if and only if $\psi_j$ is constant for all $j\in\Omega$.
Let 
$
\tilde A=(\tilde a_{jk})$ with
$\tilde a_{jk}=\frac{L_{jk}\bar \psi_j\bar\psi_k }{{|\Omega_j|}}$,
and let  $\mathcal L$ be  the  column Laplacian matrix associated with $\tilde A$; that is, $\mathcal L_{jk}=-\tilde a_{jk}$ for $j\neq k$ and $\mathcal L_{jj}=\sum_{k\not = j}\tilde a_{kj}$. Denote $\theta_j$ be the cofactor of the $j$-th diagonal element of $\mathcal L$, and consequently $(\theta_1,\dots,\theta_n)$ is an  eigenvector of $\mathcal L$ corresponding to eigenvalue 0.
Multiplying \eqref{eigenad} by $\theta_j$, summing them over all $j$, and using the Tree-Cycle identity \cite[Theorem 2.2]{LiShuai2}, we have
\begin{equation}\label{5.12}
\begin{split}
\lambda_1' \sum_{j=1}^n\theta_j \bar{\psi_j^2}\le&d_I\sum_{j=1}^n\sum_{k\neq j} \f{\theta_jL_{kj}\bar\psi_k\bar\psi_j}{|\Omega_k|}\left(\ds\frac{\bar\psi_k'}{\psi_k}-\ds\frac{\bar\psi_j'}{\psi_j}\right)+\sum_{j=1}^n\sum_{k\neq j} \ds\frac{\theta_jL_{kj}\bar \psi_j\bar\psi_k }{{|\Omega_k|}}\left(1-\ds\frac{w_j}{w_k}\right)\\
=&\sum_{\mathcal{Q} \in \mathbb{Q}} w(\mathcal{Q}) \sum_{(s,r)\in E(\mathcal{C}_{\mathcal{Q}})} \Big(\frac{\bar\psi_r'}{\bar\psi_r}-\frac{\bar\psi_s'}{\bar\psi_s}+1-\frac{w_s}{w_r}\Big),
\end{split}
\end{equation}
where  $\mathbb{Q}$ is the set of all spanning unicycle graphs of $(\mathcal{G}, \tilde{A})$, $w(\mathcal{Q})>0$ is the weight of $\mathcal{Q}$, and $\mathcal{C}_\mathcal{Q}$ denotes the directed cycle of $\mathcal{Q}$ with directed edge set $E(\mathcal{C}_{\mathcal{Q}})$. 
Along any directed cycle $\mathcal{C}_\mathcal{Q}$ of length $l$,
\begin{equation}\label{ineq}
\sum_{(s,r)\in E(\mathcal{C}_{\mathcal{Q}})} \Big(1-\frac{w_s}{w_r}\Big)  \le l - l\cdot \Big(\prod_{(s,r)\in E(\mathcal{C}_{\mathcal{Q}})} \frac{w_s}{w_r}\Big)^{1/l}=l-l\cdot 1 = 0,
\end{equation}
 and 
 \begin{equation}\label{eqad}
\sum_{(s,r)\in E(\mathcal{C}_{\mathcal{Q}})} \Big(\frac{\bar\psi_r'}{\bar\psi_r}-\frac{\bar\psi_s'}{\bar\psi_s}\Big) = 0,
\end{equation}
where the equality in \eqref{ineq} holds if and only if $w_s=w_r$ for any $(s,r)\in E(\mathcal{C}_{\mathcal{Q}})$.

Combining \eqref{5.12}-\eqref{eqad}, we have
$\la'_0\le0$. Moreover if $\la_1'=0$ for some $d_I>0$,  $\psi_j$ is constant for all $j\in\Omega$ and $\psi_1=\dots=\psi_n$. Then plugging into the first equation of \eqref{eigaad},   we have $a\beta_j(x)-\gamma_j(x)=C$ for some constant $C$ for all $x\in\bar\Omega_j$ and $j\in\Omega$. In this case, $\lambda_1=C$ for all $d_I>0$.

(2) Fix $j_0\in\Omega$.
Using similar arguments as in (1), we compute that
\begin{equation}\label{djd}
\lambda_1' \sum_{j=1}^n\theta_j \bar{\psi_j^2}=d_I\sum_{j=1}^n\sum_{k\neq j} \f{\theta_jL_{kj}\bar\psi_k\bar\psi_j}{|\Omega_k|}\left(\ds\frac{\bar\psi_k'}{\psi_k}-\ds\frac{\bar\psi_j'}{\psi_j}\right)-\theta_{ j_0}\int_{\Omega_{ j_0}}|\triangledown\psi_{j_0}|^2dx,
\end{equation}
where $'$ denotes $\partial/\partial d_{I_{ j_0}}$ and $(\theta_1,\dots,\theta_n)$ is defined as in (1).
Then by \eqref{eqad} and the Tree-Cycle identity again, we have
\begin{equation}\label{tran2}
\sum_{j=1}^n\sum_{k\neq j} \f{\theta_jL_{kj}\bar\psi_k\bar\psi_j}{|\Omega_k|}\left(\ds\frac{\bar\psi_k'}{\psi_k}-\ds\frac{\bar\psi_j'}{\psi_j}\right)=0,
\end{equation}
which implies that $\la_1'\le0$ and $\la_1$ is decreasing in $d_{I_0}$.

(3) 
Using similar arguments as in (1), we can show
\begin{equation}\label{eigint}
   \lambda_1'\sum_{j=1}^n \theta_j \bar{\psi_j^2}=d_I\sum_{j=1}^n\sum_{k\neq j} \f{\theta_jL_{kj}\bar\psi_k\bar\psi_j}{|\Omega_k|}\left(\ds\frac{\bar\psi_k'}{\psi_k}-\ds\frac{\bar\psi_j'}{\psi_j}\right) +\sum_{j=1}^n\overline{\beta_j\psi_j^2},
\end{equation}
where $'$ denotes $\partial/\partial a$ and $(\theta_1,\dots,\theta_n)$ is defined as in (1). So by \eqref{tran2},  $\lambda_1'>0$ for any $a>0$ and  $\lambda_1$ is strictly increasing in $a$.
\end{proof}

 


Then we study the profiles of $\la_1$ when the  movement
rates are small or large.
\begin{lemma}\label{lamma_lab0}
Suppose that \emph{(A1)-(A2)} hold. Let  $\la_1(d_I,d_{I_1},\dots,d_{I_n}, a)$ be the principal eigenvalue of \eqref{eiga}. 
Then 
\begin{equation}\label{lim0}
    \lim_{\max\{d_I, d_{I_1}, ..., d_{I_n}\}\to 0} \lambda_1=\max_{x\in\overline\Omega_j\\ j\in\Omega} \{a\beta_j(x)-\gamma_j(x)\}, 
\end{equation}
and
\begin{equation}\label{liminf}
    \lim_{\min\{d_I, d_{I_1}, ..., d_{I_n}\}\to \infty} \lambda_1=\sum_{j\in\Omega}\alpha_j \frac{a\bar\beta_j-\bar\gamma_j}{|\Omega_j|},
\end{equation}
where $\bm \alpha=(\alpha_1,\dots,\alpha_n)$ is the right eigenvector of $L$ corresponding to eigenvalue $0$ with $\sum_{j\in\Omega}\alpha_j=1$.
\end{lemma}
\begin{proof}
We introduce a parameter $q\in (0, \infty)$ and denote 
$$
\la_1(q,a):=\la_1(qd_I,qd_{I_1},\dots,qd_{I_n}, a).
$$

By Lemma \ref{lemma_a}, 
it suffices to show that $\lambda_1(q, a)$ satisfies
\begin{equation}\label{suppl1}
    \lim_{q\to 0} \lambda_1(q, a)=\max_{x\in\overline\Omega_j\\ j\in\Omega} \{a\beta_j(x)-\gamma_j(x)\},
\end{equation}
and
\begin{equation}\label{suppl2}
\;\;\lim_{q\to \infty} \lambda_1(q, a)=\sum_{j\in\Omega}\alpha_j \frac{a\bar\beta_j-\bar\gamma_j}{|\Omega_j|}.
\end{equation}
First we prove \eqref{suppl1}. Let $\bm\varphi=(\varphi_1, ..., \varphi_n)$ be a positive eigenvector corresponding with $\lambda_1(q, a)$. We have 
$$
\lambda_1(q, a)\varphi_j\ge qd_{I_j}\Delta \varphi_j-qd_I\sum_{k\neq j} L_{kj}  \varphi_j+(a\beta_j-\gamma_j)  \varphi_j, \ \   j\in\Omega.
$$
Let $\lambda^j_1(q, a)$, $j\in\Omega$, be the principal eigenvalue of the following eigenvalue problem
\begin{equation}\label{eigcom}
\begin{cases}
\lambda\phi_j= qd_{I_j}\Delta \phi_j-qd_I\sum_{k\neq j} L_{kj}  \phi_j+(a\beta_j-\gamma_j)  \phi_j, \ \ &x\in\Omega_j, \\
\ds\f{\partial \phi_j}{\partial \nu}=0, \ \ &x\in\partial\Omega_j.
\end{cases}
\end{equation}
By the ``min-max" formula (\cite{berestycki1994principal}) of $\lambda^j_1(q, a)$, we have 
\begin{eqnarray*}
\lambda^j_1(q, a)&=&\min_{\phi\in W_j} \max_{x\in\Omega_j} 
\frac{qd_{I_j}\Delta \phi-qd_I\sum_{k\neq j} L_{kj}  \phi+(a\beta_j-\gamma_j)  \phi}{\phi}\\
&\le& \max_{x\in\Omega_j} 
\frac{qd_{I_j}\Delta \varphi_j-qd_I\sum_{k\neq j} L_{kj}  \varphi_j+(a\beta_j-\gamma_j)  \varphi_j}{\varphi_j}\le\lambda_1(q, a),
\end{eqnarray*}
where 
$$
W_j:=\left\{u\in C^2(\Omega_j)\cap C^1(\bar\Omega_j): \ \frac{\partial u}{\partial \nu}=0  \ \text{on } \partial\Omega_j \ \ \text{and} \  \ u(x)>0 \text{ for all}\ x\in\Omega_j\right\}.
$$
Then we have $\lambda_1(q, a)\ge \lambda^j_1(q, a)$ for all $j\in\Omega$. It is well-known \cite{cantrell2004spatial} that 
$$
\lim_{q\to 0} \lambda_1^j(q, a)= \max_{x\in\bar\Omega_j} \{a\beta_j(x)-\gamma(x)\}.
$$
It follows that 
\begin{equation}\label{eigliminf}
   \liminf_{q\to 0} \lambda_1(q, a)\ge  \max_{x\in\bar\Omega_j, j\in\Omega} \{a\beta_j(x)-\gamma(x)\}.
\end{equation}
On the other hand, integrating the first equation of \eqref{eiga} over $\Omega_j$ and summing up over $j\in\Omega$, we obtain 
$$
\lambda_1(q, a)\sum_{j\in\Omega} \bar\varphi_j=\sum_{j\in\Omega} \int_{\Omega_j} (a\beta_j-\gamma_j)\varphi_j,
$$
which leads to
\begin{equation}\label{lambda_bound}
\lambda_1(q, a) \le  \max_{x\in\bar\Omega_j, j\in\Omega} \{a\beta_j(x)-\gamma(x)\}
\end{equation}
for all $q, a>0$. This combined with \eqref{eigliminf} proves \eqref{lim0}.

Then we consider \eqref{suppl2}. We normalize $\bm\varphi$ such that $\sum_{j\in\Omega} \|\varphi_j\|_{L^\infty(\Omega_j)}=1$. By \eqref{eiga}, \eqref{lambda_bound} and the parabolic estimate, we know that $\varphi_j$ is uniformly bounded in $W^{2, p}(\Omega_j)$ for all $j\in \Omega_j$ and $q>1$. Therefore, restrict to a subsequence if necessary we have $\varphi_j\to \varphi_j^*$ weakly in $W^{2, p}(\Omega_j)$ as $q\to\infty$ for some $\varphi_j^*\in W^{2, p}(\Omega_j)$ with $\sum_{j\in\Omega} \|\varphi_j^*\|_{L^\infty(\Omega_j)}=1$. Dividing \eqref{eiga} by $q$ and taking $q\to\infty$, we find 
\begin{equation}\label{eig22}
\begin{cases}
\ds d_{I_j}\Delta \varphi_j^*+d_I\sum_{k\neq j} \left(\f{L_{jk}\bar \varphi_k^*}{|\Omega_j|}-  L_{kj}  \varphi_j^*\right)=0, &x\in\Omega_j,  j\in\Omega,\\
\ds\f{\partial \varphi_j^*}{\partial \nu}=0, &x\in\partial\Omega_j,  j\in\Omega.
\end{cases}
\end{equation}
Similar to the proof of Lemma \ref{lemma_DFE}, we can show $\bm\varphi^*=k^*(\alpha_1/{|\Omega_1|}, ..., \alpha_n/{|\Omega_n|})$, where 
$$
k^*=\frac{1}{\sum_{j\in\Omega} \alpha_j/|\Omega_j|}.
$$
Integrating the first equation of \eqref{eiga} over $\Omega_j$ and summing up the equations over $j\in\Omega$, we obtain 
$$
\lambda_1(q, a) \sum_{j\in\Omega}\int_{\Omega_j}\varphi_j =\sum_{j\in\Omega}\int_{\Omega_j}(a\beta_j-\gamma_j)\varphi_j .
$$
Taking $q\to\infty$ and plugging in $\bm\varphi^*$, we obtain \eqref{liminf}.
\end{proof}

Combining Lemmas \ref{lemma_a} and \ref{lamma_lab0}, we obtain the following main result about $\mathcal{R}_0$:
\begin{theorem}
Suppose that \emph{(A1)-(A2)} hold. Then the following statements hold:
\begin{enumerate}
\item [\emph{(1)}] $\mathcal R_0$ is 
decreasing in  $d_I\in(0, \infty)$. Moreover, $\mathcal R_0$  is strictly decreasing in $d_I$ if and only if $(\beta_1, \dots, \beta_n)$ is not a multiple of $(\gamma_1, \dots, \gamma_n)$.
\item [\emph{(2)}]$\mathcal R_0$ is 
decreasing in $d_{I_j}\in (0, \infty)$ for each $j\in\Omega$.
\item[\emph{(3)}]  $\mathcal R_0$ satisfies
    $$
\ds\lim_{\max\{d_I, d_{I_1}, ..., d_{I_n}\}\to 0} \mathcal{R}_0= \sup_{x\in\bar\Omega_j, j\in\Omega} \frac{\beta_j(x)}{\gamma_j(x)},$$ and $$\ds    \lim_{\min\{d_I, d_{I_1}, ..., d_{I_n}\}\to \infty} \mathcal{R}_0=  \frac{\sum_{j\in\Omega}\alpha_j\frac{\bar\beta_j}{|\Omega_j|}}{\sum_{j\in\Omega}\alpha_j\frac{\bar\gamma_j}{|\Omega_j|}},
$$
where $\bm \alpha=(\alpha_1,\dots,\alpha_n)$ is the right eigenvector of $L$ corresponding to eigenvalue $0$ with $\sum_{j\in\Omega}\alpha_j=1$.
\end{enumerate}
\end{theorem}
\begin{proof}
We only prove (1) as (2) can be proved similarly using Lemma \ref{lemma_a}.
Note that $\tilde\lambda_1:=1/\mathcal{R}_0$ is a principal eigenvalue of 
\begin{equation}\label{eigd}
\begin{cases}
\ds -\lambda  \beta_j \varphi_j=d_{I_j}\Delta \varphi_j+d_I\sum_{k\neq j} \left(\f{L_{jk}\bar \varphi_k}{|\Omega_j|}-  L_{kj}  \varphi_j\right)-\gamma_j  \varphi_j,&x\in\Omega_j,  j\in\Omega,\\
\ds\f{\partial \varphi_j}{\partial \nu}=0, &x\in\partial\Omega_j,  j\in\Omega.
\end{cases}
\end{equation}
If  $(\beta_1, \dots, \beta_n)=k(\gamma_1, \dots, \gamma_n)$ for some $k>0$,  by Lemma \ref{lemma_a}(3), it is easy to see that $\tilde\lambda_1(q)=1/k$ and $\mathcal{R}_0=k$ for all $q>0$. So the claim holds in this situation. Hence, we may assume  that $(\beta_1, \dots, \beta_n)$ is not a multiple of $(\gamma_1, \dots, \gamma_n)$. We consider two cases:

\noindent\emph{Case 1}. For any $a>0$, there does not exist a constant $C$ such that $a\beta_j(x)-\gamma_j(x)=C$ for all $x\in\bar\Omega_j$ and $j\in\Omega$. 

By Lemma \ref{lemma_a}, $\lambda_1$ is strictly decreasing in $d_I$ for any $a>0$. Let $d^{1}_I>d^{1}_I>0$. Then, we have
 $$
 \lambda_1\left(d^{2}_I,\tilde\lambda_1\left(d^{1}_I\right)\right)> \lambda_1\left(d^{1}_I,\tilde\lambda_1\left(d^{1}_I\right)\right)= \lambda_1\left(d^{2}_I,\tilde\lambda_1\left(d^{2}_I\right)\right)=0.
 $$
 By Lemma \ref{lemma_a}, we have $\tilde\lambda_1\left(d^{1}_I\right)>\tilde\lambda_1\left(d^{2}_I\right)$. Therefore, $\tilde\lambda_1$ is strictly increasing in $d_I$, and $\mathcal{R}_0$ is strictly decreasing in $d_I$. 
 
 \noindent\emph{Case 2}. There exists $\tilde a>0$ such that  $\tilde a\beta_j(x)-\gamma_j(x)=C$ for all $x\in\bar\Omega_j$ and $j\in\Omega$ for some constant $C$. 
 
 Since $(\beta_1, \dots, \beta_n)$ is not a multiple of $(\gamma_1, \dots, \gamma_n)$, $\tilde a$ is the unique number such that $(\tilde a\beta_1-\gamma_1,\dots, \tilde a\beta_n-\gamma_n)$ is a multiple of $(1, \dots, 1)$ and $C\neq 0$. 
  If $C>0$, then 
 $$
 \inf_{x\in\bar\Omega} \frac{\beta_j(x)}{\gamma_j(x)}>\frac{1}{\tilde a}, \ \ \text{for all } j\in\Omega.
 $$
 By Lemma \ref{lemma_Rb}, we have $\mathcal{R}_0>1/\tilde a$ and $\tilde\lambda_1<\tilde a$ for all $d_I>0$.  By   Lemma \ref{lemma_a}, $\lambda_1$ is strictly decreasing in $q$ for any $a<\tilde a$. Then similar to Case 1, we can show that $\tilde\lambda_1$ is strictly increasing in $d_I$, and $\mathcal{R}_0$ is strictly decreasing in $d_I$. 
 
 If $C<0$, then 
 $$
 \sup_{x\in\bar\Omega} \frac{\beta_j(x)}{\gamma_j(x)}<\frac{1}{\tilde a}, \ \ \text{for all } j\in\Omega.
 $$
 By Lemma \ref{lemma_Rb}, we have $\mathcal{R}_0<1/\tilde a$ and $\tilde\lambda_1>\tilde a$ for all $d_I>0$.  Then we can prove the monotonicity of $\tilde\lambda_1$ and $\mathcal{R}_0$ with respect to $d_I$ similar to the case $C>0$.

(3) We introduce the parameter $q\in (0, \infty)$ again and denote 
$$\mathcal R_0(q):=\mathcal \la_1(qd_I,qd_{I_1},\dots,qd_{I_n}).$$
By (1)-(2), it suffices to compute the limits of $\mathcal{R}_0(q)$ as $q\to 0$ or $\infty$. 
 Since $\mathcal{R}_0(q)$ is decreasing in  $q$, we may assume 
 $$
 \lim_{q\to 0} \mathcal{R}_0(q)=1/l_1\ \ \text{and}\ \   \lim_{q\to \infty} \mathcal{R}_0(q)=1/l_2
 $$
 for some $l_1 \in [0, \infty)$ and $l_2\in (0, \infty]$. Therefore, $\lim_{q\to 0} \tilde \lambda_1(q)=l_1$ and $ \lim_{q\to \infty} \tilde \lambda_1(q)=l_2$. We first suppose $l_1\neq 0$ and $l_2\neq \infty$.  Let $\epsilon>0$ be given. Then there exists $\delta>0$ such that 
 $$
l_1-\epsilon< \tilde \lambda_1(q)<l_1+\epsilon, \ \ \text{for any } q<\delta,
 $$
 and 
  $$
l_2-\epsilon< \tilde \lambda_1(q)<l_2+\epsilon, \ \ \text{for any } q>\frac{1}{\delta}. 
 $$
 By the monotonicity of $\lambda_1(q, a)$ in $q$, we have 
 $$
\lambda_1(q, l_1-\epsilon)< \lambda_1(q, \tilde\lambda_1(q))=0<  \lambda_1(q, l_1+\epsilon), \ \ \text{for any } q<\delta,
 $$
 and 
  $$
\lambda_1(q, l_2-\epsilon)< \lambda_1(q, \tilde\lambda_1(q))=0<  \lambda_1(q, l_2+\epsilon), \ \ \text{for any } q>\frac{1}{\delta}.
 $$
 By Lemma \ref{lamma_lab0}, we have 
\begin{equation}\label{l1b}
  \max_{x\in\overline\Omega_j\\ j\in\Omega} \{(l_1-\epsilon)\beta_j(x)-\gamma_j(x)\}
\le 0\le \max_{x\in\overline\Omega_j\\ j\in\Omega} \{(l_1+\epsilon)\beta_j(x)-\gamma_j(x)\},
\end{equation}
 and 
  $$
  \sum_{j\in\Omega}\alpha_j \frac{(l_2-\epsilon)\bar\beta_j-\bar\gamma_j}{|\Omega_j|}
\le 0\le \sum_{j\in\Omega}\alpha_j \frac{(l_2+\epsilon)\bar\beta_j-\bar\gamma_j}{|\Omega_j|}.
 $$
 Taking $\epsilon\to 0$, we obtain 
 $$
  \max_{x\in\overline\Omega_j\\ j\in\Omega} \{l_1\beta_j(x)-\gamma_j(x)\} 
= 0 \ \ \text{and}\ \   \sum_{j\in\Omega}\alpha_j \frac{l_2\bar\beta_j-\bar\gamma_j}{|\Omega_j|}=0.
 $$
 Therefore, we have
 $$
 \frac{1}{l_1}
= \sup_{x\in\bar\Omega_j, j\in\Omega} \frac{\beta_j(x)}{\gamma_j(x)} \ \ \text{and}\ \   \frac{1}{l_2}=  \frac{\sum_{j\in\Omega}\alpha_j\frac{\bar\beta_j}{|\Omega_j|}}{\sum_{j\in\Omega}\alpha_j\frac{\bar\gamma_j}{|\Omega_j|}}.
 $$

 If $l_1=0$, it suffices to show that we must have $M:=\sup_{x\in\bar\Omega} {\beta_j(x)}/{\gamma_j(x)}=\infty$ in this case. Indeed, suppose to the contrary that $M<\infty$. Then we must have 
 \begin{equation}\label{MM}
     \frac{1}{M}\beta_j(x)-\gamma_j(x)\le 0, \ \ \text{for all} \ x\in\Omega_j, j\in\Omega. 
 \end{equation}
It is easy to check that the second inequality in \eqref{l1b} still holds when $l_1=0$. Taking $\epsilon\to0$, we obtain 
$$
0 \le \max_{x\in\overline\Omega_j\\ j\in\Omega} \{l_1\beta_j(x)-\gamma_j(x)\}.
$$
By \eqref{MM}, we have $l_1\ge 1/M>0$, which is a contradiction. 

Finally, we show that $l_2=\infty$ is not possible. Suppose to the contrary that $l_2=\infty$. Fix $
\tilde M>0$ such that 
$$
\sum_{j\in\Omega}\alpha_j \frac{\tilde M\bar\beta_j-\bar\gamma_j}{|\Omega_j|}>0.
$$
Since $\lim_{q\to \infty} \tilde \lambda_1(q)=l_2=\infty$, there exists $\hat q>0$ such that $\tilde \lambda_1(q)>\tilde M$ for all $q\ge \hat q$. 
 By the monotonicity of $\lambda_1(q, a)$ in $a$ (Lemma \ref{lemma_a}), we have 
 $$
\lambda_1(q, \tilde\lambda_1(q))=0>  \lambda_1(q, \tilde M), \ \ \text{for all } q\ge \hat q.
 $$
 Taking $q\to\infty$ and by Lemma \ref{lamma_lab0}, we obtain 
 $$
0\ge \lim_{q\to\infty}  \lambda_1(q, \tilde M)=\sum_{j\in\Omega}\alpha_j \frac{\tilde M\bar\beta_j-\bar\gamma_j}{|\Omega_j|},
 $$
 which is a contradiction. 
\end{proof}


\section{Numerical simulations}
In this section, we explore the impact of controlling population movement on the transmission of the disease using model \eqref{patc3}-\eqref{total}. -eps-converted-to.pdf

Let $\Omega=\{1, 2\}$, $\Omega_1=(0, 1)$ and $\Omega_2=(2, 3)$, which means that the individuals are living in two patches and each patch is an interval of length one. The transmission  and recovery rates are $\beta_1=1.5+\sin(2\pi x), \beta_2=1+\sin(2\pi x)$ and $\gamma_1=\gamma_2=1$, respectively.  The movement pattern of individuals between the two patches is described by matrix $L=(L_{jk})$ with $L_{11}=L_{22}=-1$ and  $L_{12}=L_{21}=1$, and the movement rates of susceptible and infected individuals between patches are $d_{S}=d_{I}=1$. The movement rates of individuals inside patches are $d_{S_1}=d_{S_2}=d_{I_1}=d_{I_2}=1$. The initial conditions are $S_1(x, 0)=S_2(x, 0)=5$ and $I_1(x, 0)=I_2(x, 0)=1$. We solve the model numerically and graph the total infected individuals $\int_{\Omega_1} I_1$ and $\int_{\Omega_1} I_2$ in Fig. \ref{fig1}. The figure indicates that the disease will persist in both patches.

\begin{figure}[htbp]
\centering\includegraphics[width=0.5\textwidth]{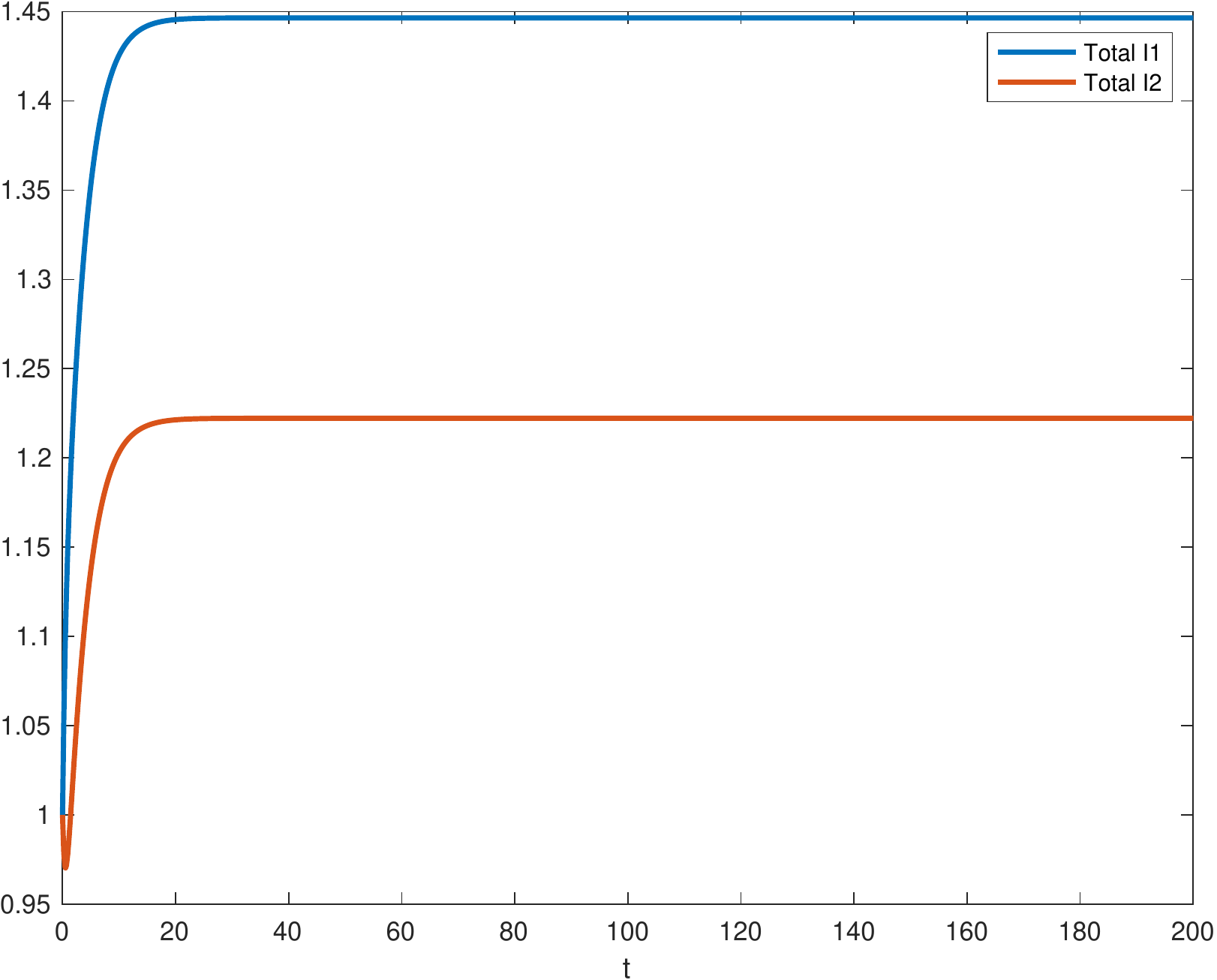}
\caption{Total infected individuals $\int_{\Omega_1} I_1$ and $\int_{\Omega_1} I_2$. The movement rates are $d_{S_1}=d_{S_2}=d_{I_1}=d_{I_2}=d_S=d_I=1$.}
\label{fig1}
\end{figure}

In the following, we will vary some  parameter values to see the impact of limiting the movement of individuals. The parameter values not mentioned below are the same as the first simulation. 
\begin{itemize}
    \item Limiting $d_{S_1}$. 
Since we are interested in limiting the movement of susceptible individuals in patch 1, we let $d_{S_1}=10^{-5}$ in this simulation. Firstly, we set $d_S=d_I=0$, which means that the patches are disconnected. As shown in  Fig. \ref{fig2}(a), the total infected individuals in patch 1  approaches zero. This is in agreement with a result in \cite{Allen}, which states that if $n=1$ limiting the movement of susceptible individuals may eliminate the disease provided that $\beta_1-\gamma_1$ changes sign in $\Omega_1$. Then we exam the impact of the movement between patches by setting $d_S=d_I=0.2$ (Fig. \ref{fig2}(b)) or $d_S=d_I=1$ (Fig. \ref{fig2}(c)).
It turns out that the infected individuals in both patches are no longer approaching zero, which implies that limiting $d_{S_1}$ cannot eliminate the disease now.

    \begin{figure}[htbp]
\centering
 \begin{subfigure}{0.32\textwidth}
 \centering
\includegraphics[width=\textwidth]{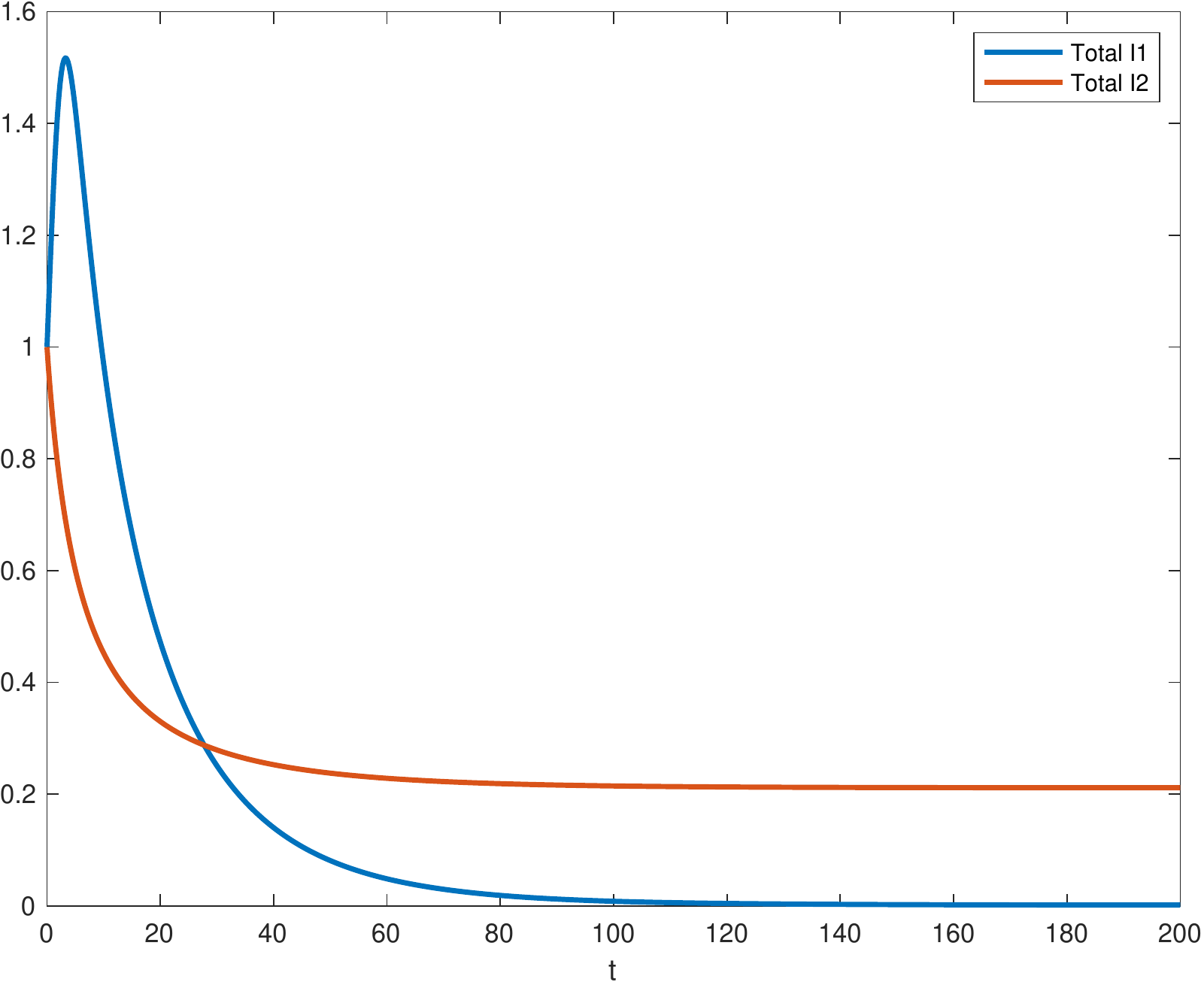}
(a)
\end{subfigure}
 \begin{subfigure}{0.32\textwidth}
 \centering
\includegraphics[width=\textwidth]{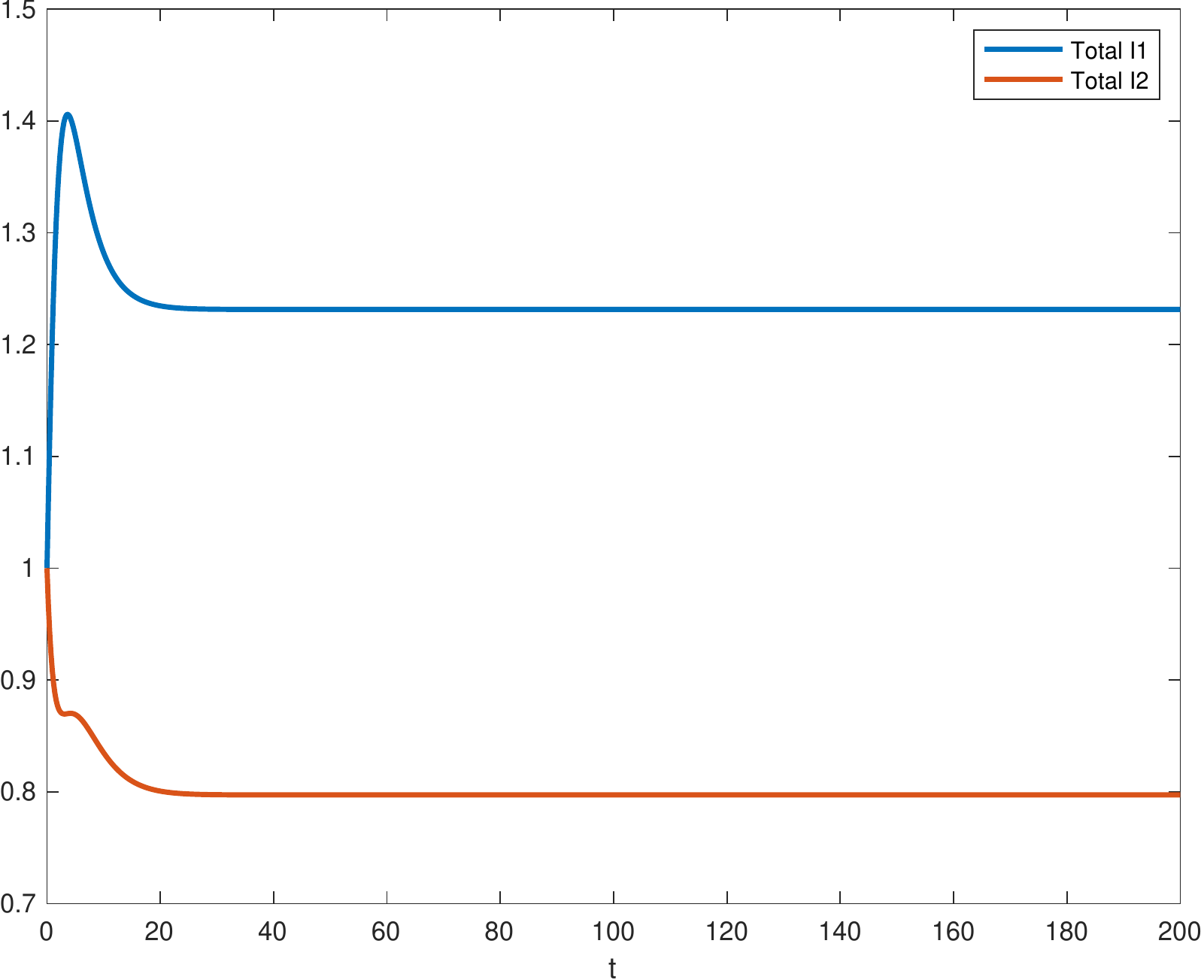}
\label{fig2b}
(b)
\end{subfigure}
 \begin{subfigure}{0.32\textwidth}
 \centering
\includegraphics[width=\textwidth]{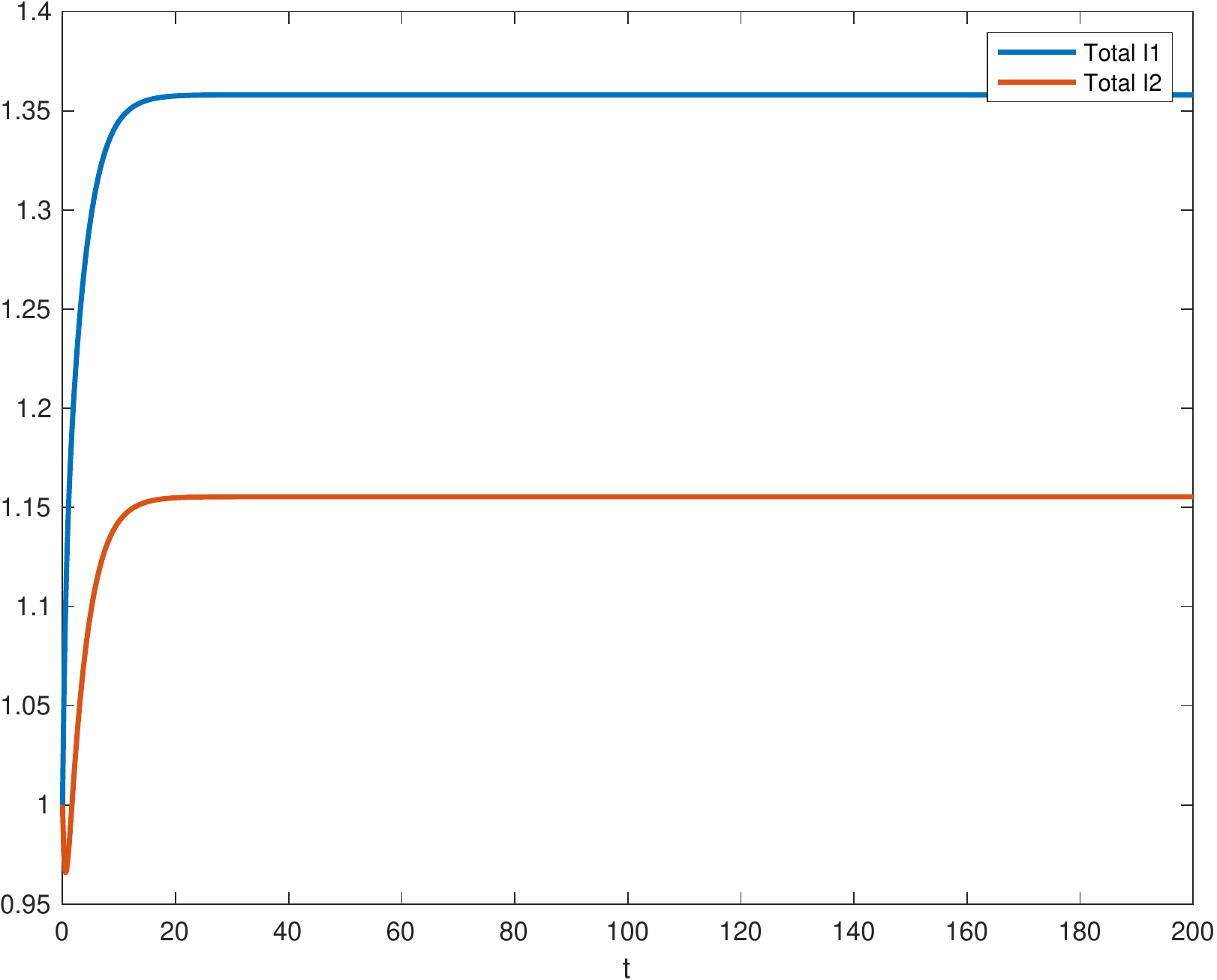}
\label{fig2c}
(c)
\end{subfigure}
\caption{Total infected individuals $\int_{\Omega_1} I_1(x, t)$ and $\int_{\Omega_1} I_2(x, t)$  with $d_{S_1}=10^{-5}$.  The left figure is for $d_S=d_I=0$; the middle figure is for $d_S=d_I=0.2$; the right figure is for $d_S=d_I=1$. The other movement rates are $d_{S_2}=d_{I_1}=d_{I_2}=1$.}
\label{fig2}
\end{figure}

    \item Limiting $d_{S_1}$ and $d_{S_2}$. We simulate the impact of limiting the movement of susceptible individuals in both patches by setting $d_{S_1}=d_{S_2}=10^{-5}$. If the patches are disconnected ($d_{S}=d_{I}=0$), as shown in the Fig. \ref{fig3}(a), the infected individuals in both patches are eliminated. However if there are movement of individuals between patches, limiting the movement of susceptible individuals cannot eliminate the disease anymore (see Fig. \ref{fig3}(b)-(c)). Moreover, larger movement rates between patches  seem to increase the epidemic size in both patches. 

    \begin{figure}[htbp]
\centering
 \begin{subfigure}{0.32\textwidth}
 \centering
\includegraphics[width=\textwidth]{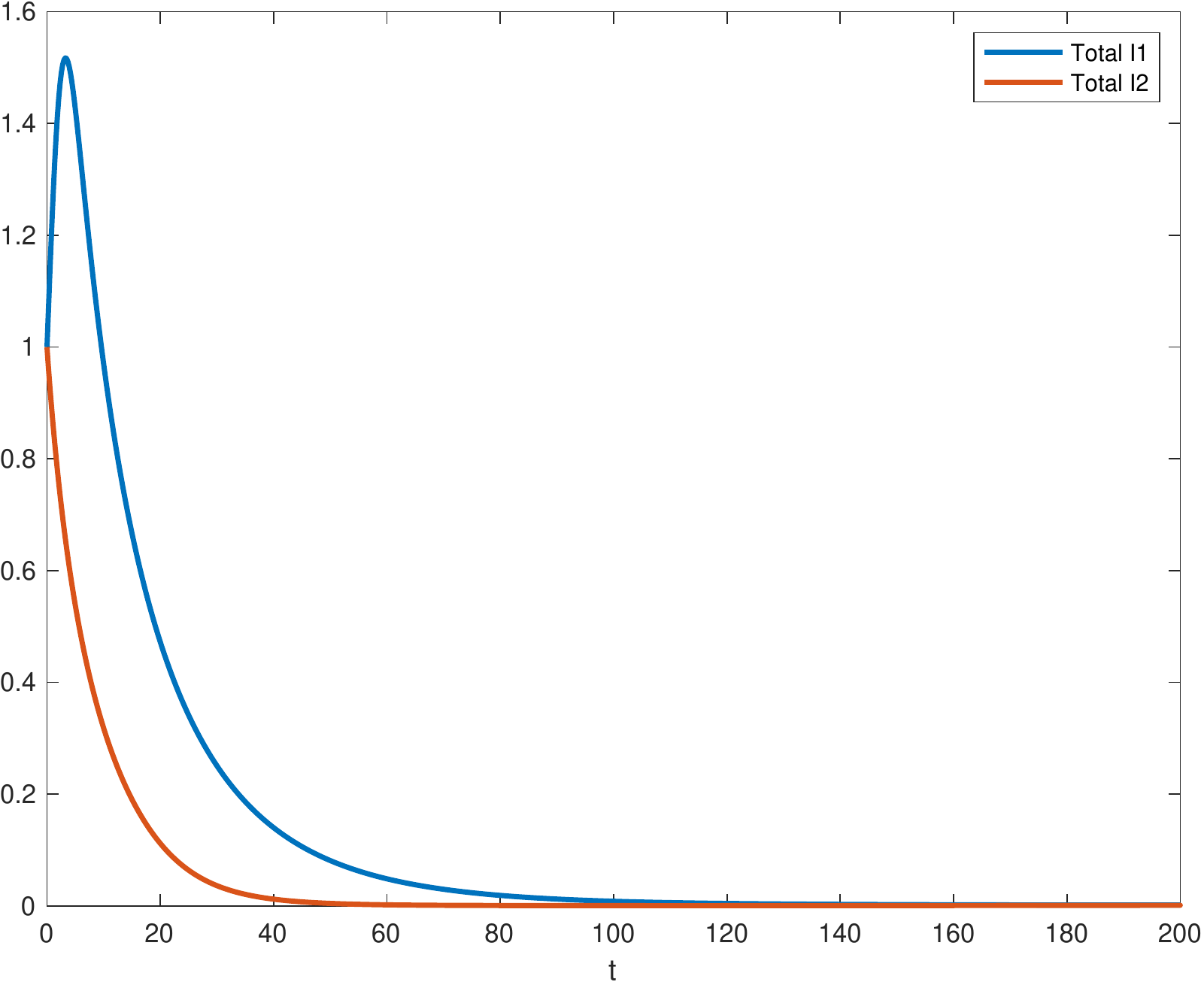}
(a)
\end{subfigure}
 \begin{subfigure}{0.32\textwidth}
 \centering
\includegraphics[width=\textwidth]{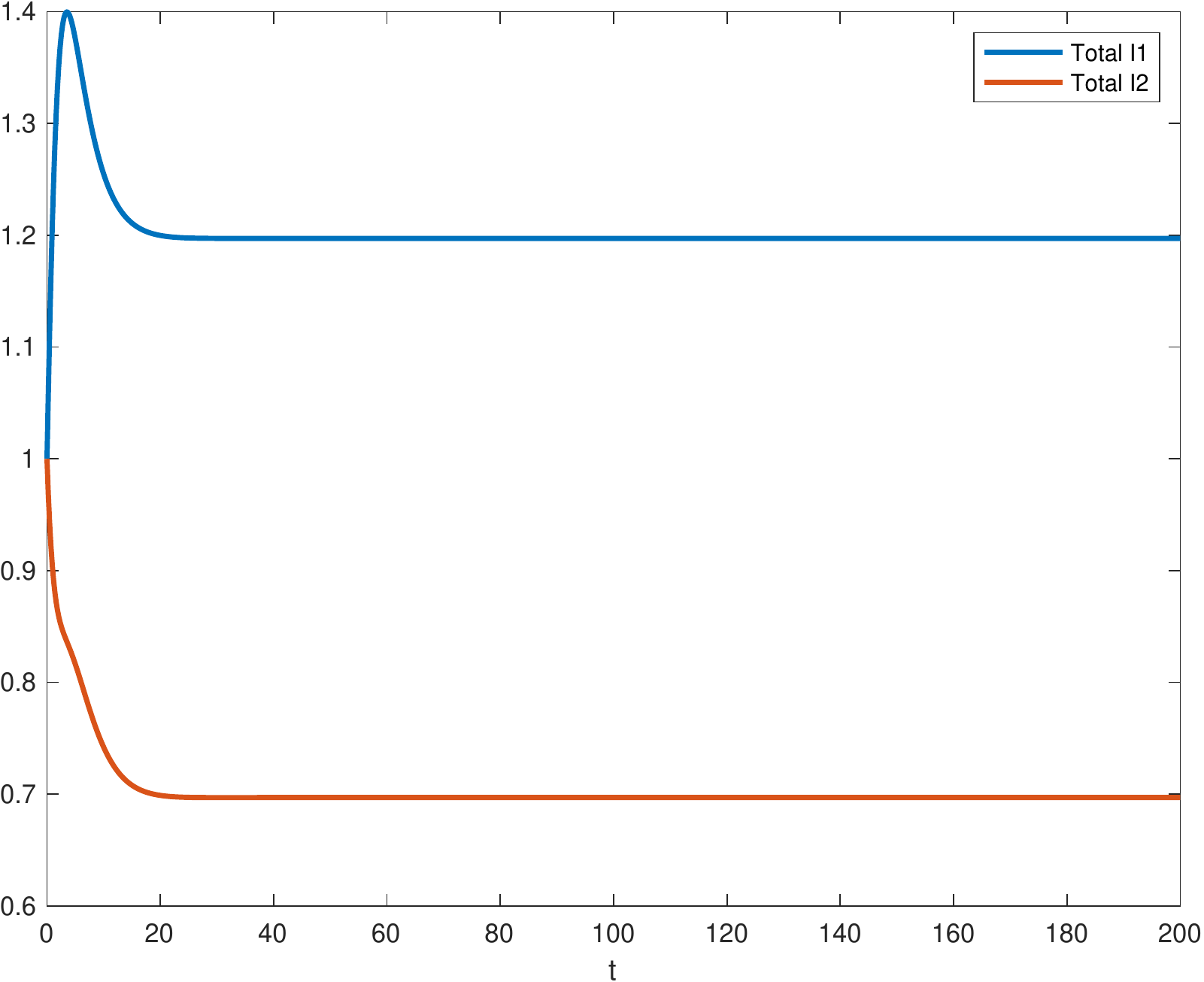}
\label{fig2b2}
(b)
\end{subfigure}
 \begin{subfigure}{0.32\textwidth}
 \centering
\includegraphics[width=\textwidth]{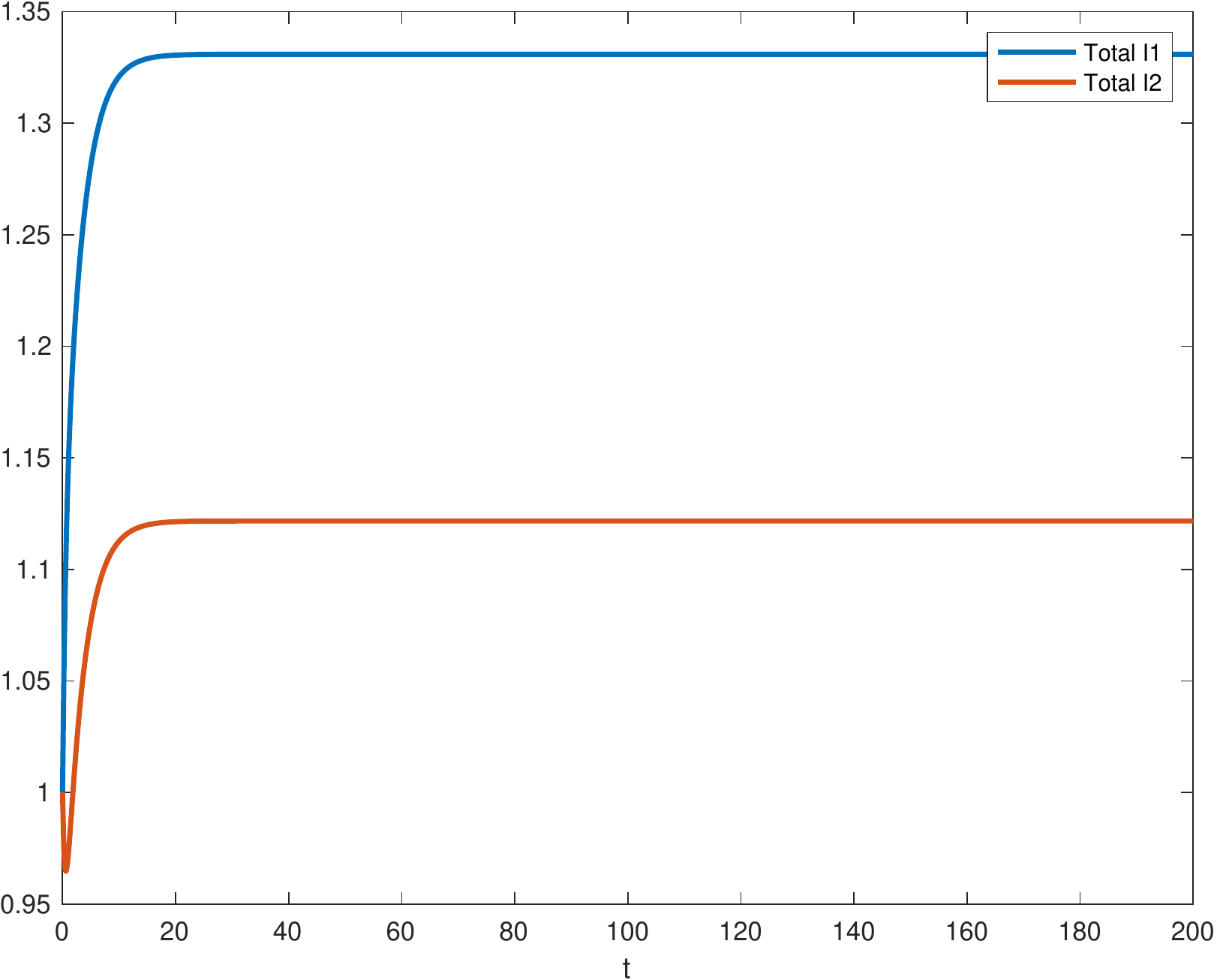}
\label{fig2c2}
(c)
\end{subfigure}
\caption{Total infected individuals $\int_{\Omega_1} I_1(x, t)$ and $\int_{\Omega_1} I_2(x, t)$  with $d_{S_1}=d_{S_2}=10^{-5}$.  The left figure is for $d_S=d_I=0$; the middle figure is for $d_S=d_I=0.2$; the right figure is for $d_S=d_I=1$. The other movement rates are $d_{I_1}=d_{I_2}=1$.}
\label{fig3}
\end{figure}

  \item Limiting $d_{I_1}$ and $d_{I_2}$. If the patches are disconnected ($d_S=d_I=0$), then limiting the movement of infected individuals cannot eliminate the disease as shown in Fig. \ref{fig4}(a). This is in agreement of the results proved in \cite{Peng2009}. Then we set $d_S=d_I=1$, which means that there are movement of individuals between patches. As expected, limiting the movement of infected individuals in one patch (see Fig. \ref{fig4}(b)) or two patches (Fig. \ref{fig4}(c)) cannot eliminate the disease. Moreover comparing  Fig. \ref{fig1} and Fig. \ref{fig4}, we can see that limiting the movement of infected individuals may even increase the epidemic size.

    \begin{figure}[htbp]
\centering
\begin{subfigure}{0.32\textwidth}
 \centering
\includegraphics[width=\textwidth]{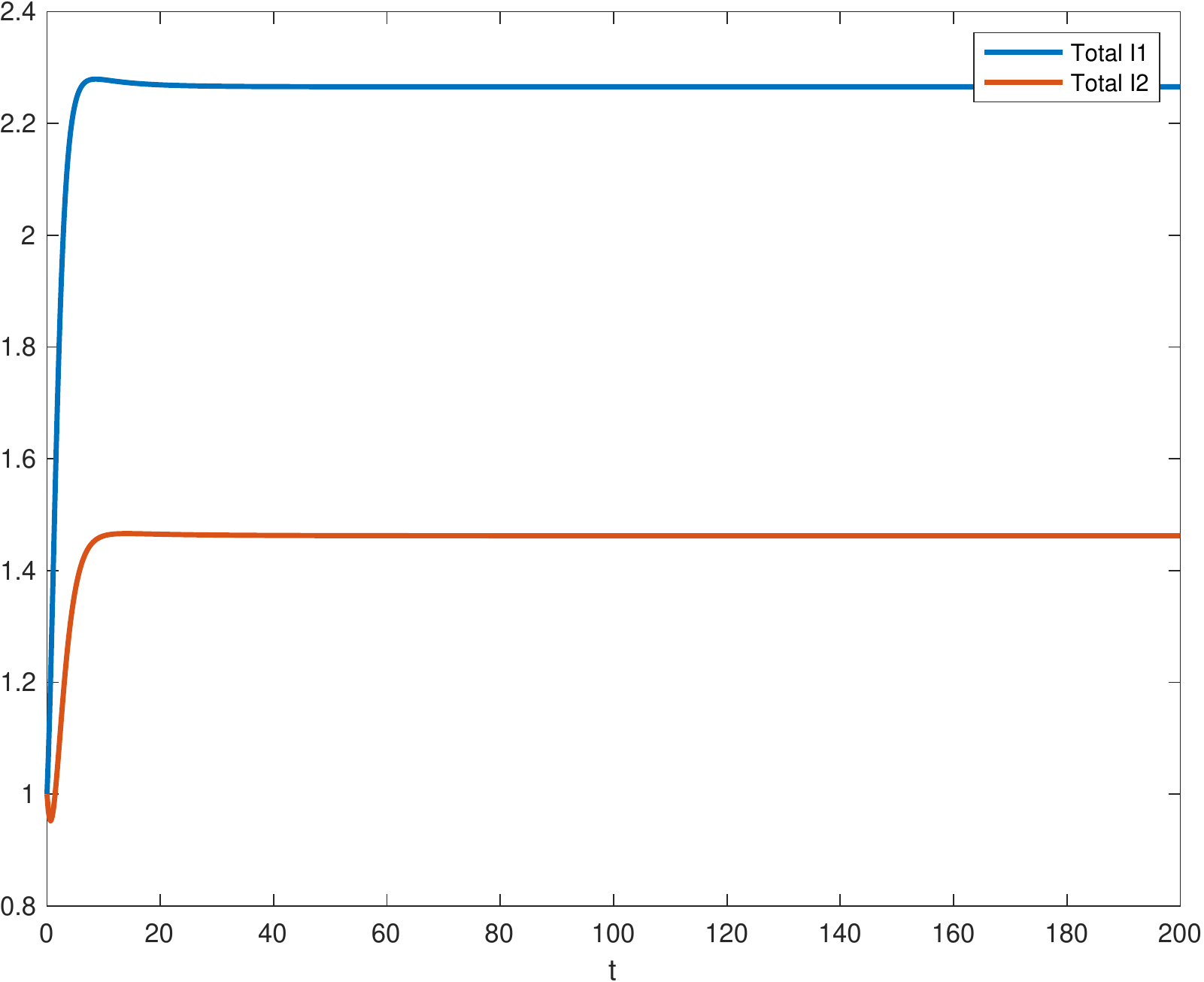}
(a)
\end{subfigure}
 \begin{subfigure}{0.32\textwidth}
 \centering
\includegraphics[width=\textwidth]{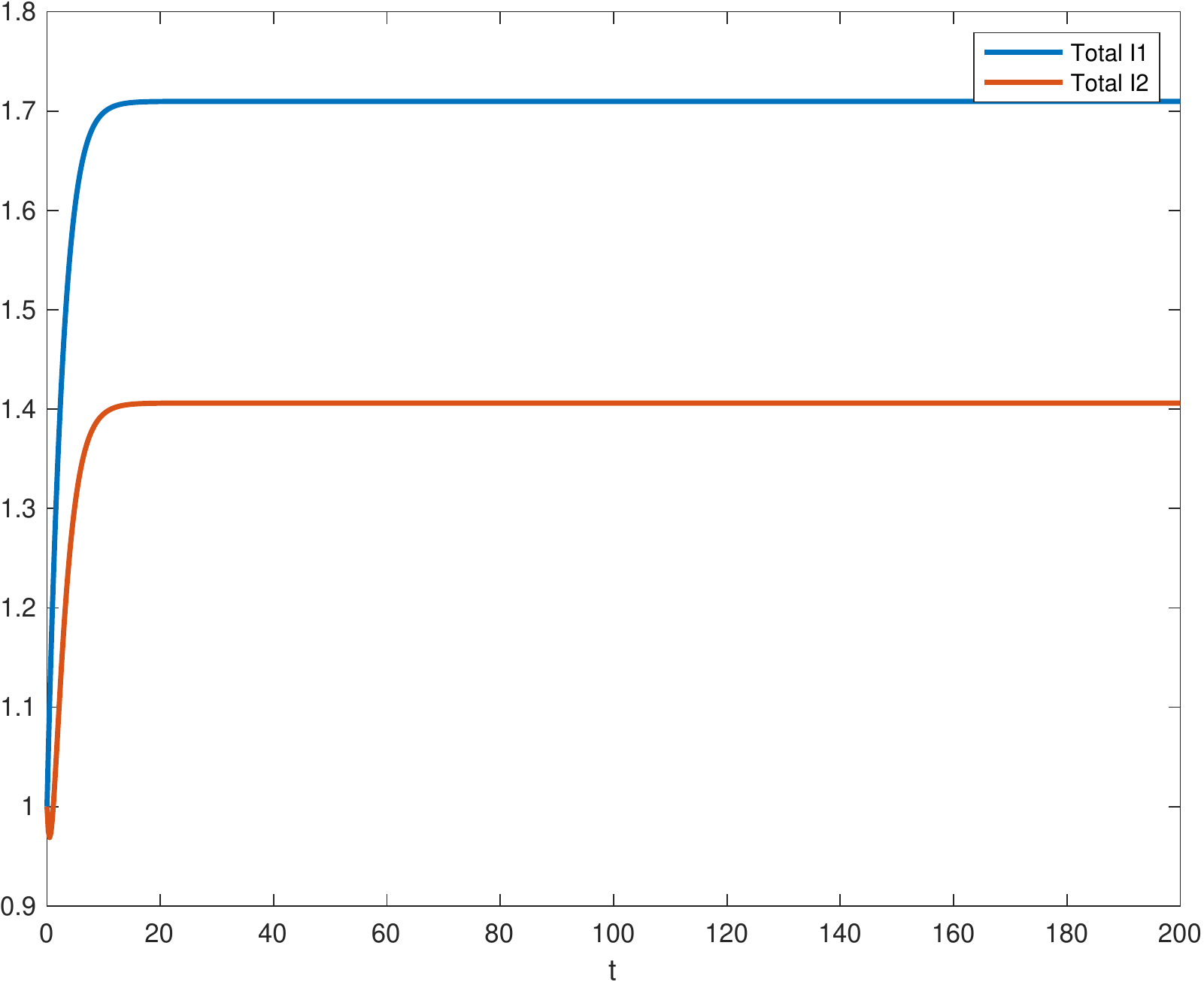}
\label{fig2b3}
(b)
\end{subfigure}
 \begin{subfigure}{0.32\textwidth}
 \centering
\includegraphics[width=\textwidth]{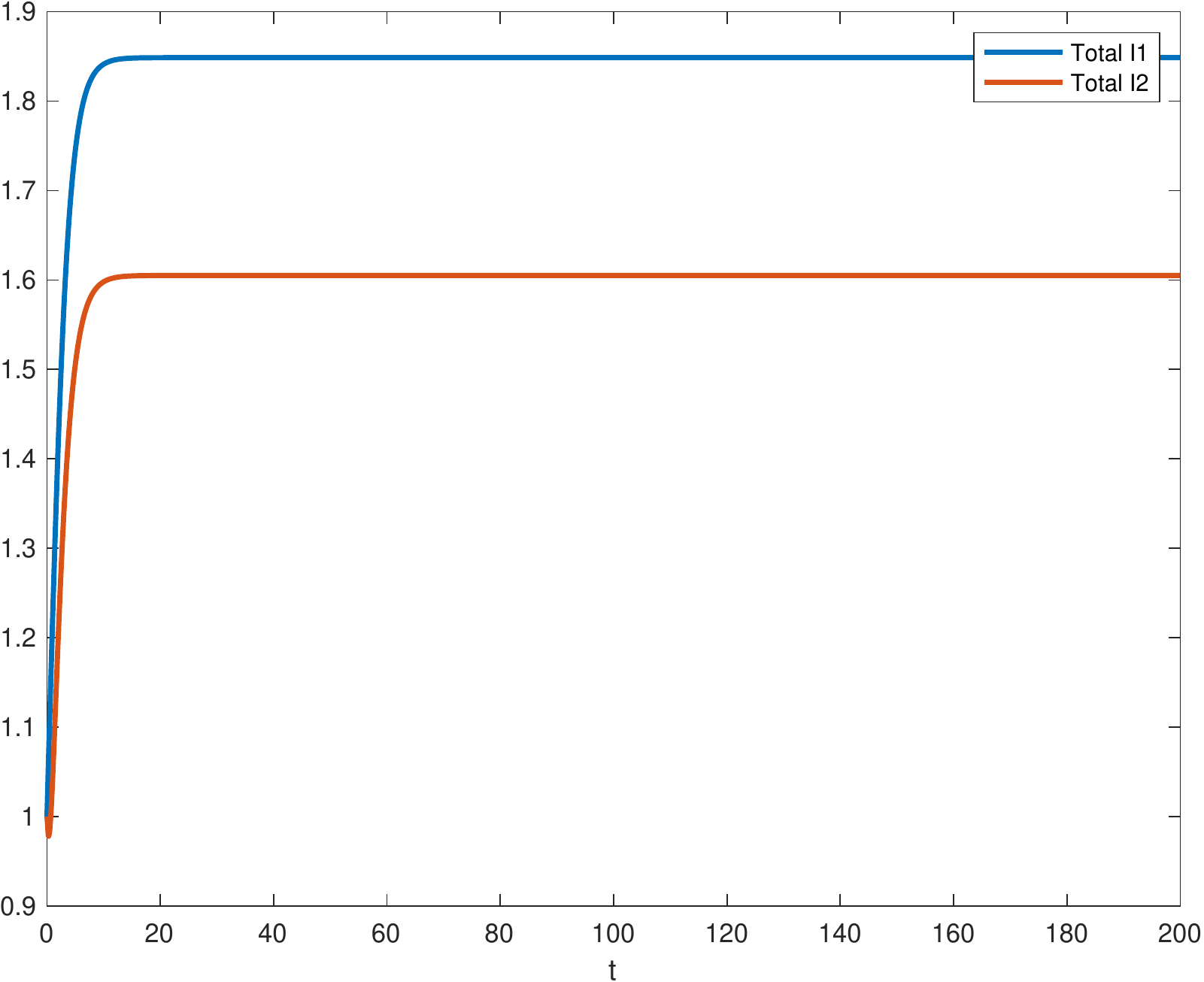}
\label{fig2c3}
(c)
\end{subfigure}
\caption{Total infected individuals $\int_{\Omega_1} I_1(x, t)$ and $\int_{\Omega_1} I_2(x, t)$.  The left figure is for  $d_{I_1}=d_{I_2}=10^{-5}$ and $d_S=d_I=0$; the middle figure is for $d_{I_1}=10^{-5}$, $d_{I_2}=1$ and $d_S=d_I=1$; the right figure is for $d_{I_1}=d_{I_2}=10^{-5}$ and $d_S=d_I=1$. }
\label{fig4}
\end{figure}

\end{itemize}

\section{Summary}

In this paper, we introduce a reaction-diffusion SIS epidemic patch model that describes the transmission of a disease among different regions. The movement of individuals inside regions is described by diffusion terms, and  the movement among regions is modeled by a matrix. We define a basic reproduction number $\mathcal{R}_0$ and show that it is a threshold value for the global dynamics of the model. In particular if $\mathcal{R}_0<1$, then the DFE is globally asymptotically stable; if $\mathcal{R}_0>1$, the solutions of the model are uniformly persistent and the model has at least one EE. If the matrix $L$ that describes the dispersal pattern among regions is symmetric,  we are able to obtain a variational formula for $\mathcal{R}_0$. From the formula, we can easily see that $\mathcal{R}_0$ is monotonely  dependent on the movement rates of the infected individuals. If $L$ is not symmetric, we use some graph-theoretic technique to show that  the monotone property still holds.  Moreover, we compute the limits of $\mathcal{R}_0$ as the movement rates approach zero or infinity, from which we can better see how $\mathcal{R}_0$ depends on the disease transmission and recovery rates. 

Theoretical investigations on the profiles of the EE as the movement rates of susceptible or infected individuals approaching zero has not been studied yet. However, simulations show that limiting the local movement of susceptible  individuals can no longer eliminate the disease when there are movement of individuals among patches. Moreover, limiting the local movement of infected individuals may even increase the epidemic size. Of course, these claims should be tested further by more rigorous theoretical analysis. Possible extensions of the model such as the introduction of demographic structure, exposed and/or recovered departments and multiple strains of the disease will be left as future works.

\section{Appendix}

\subsection{Eigenvalue problem}
Let $X$ be an ordered Banach space and $X_+$ be a nontrivial  cone of $X$ that is normal and generating. Let $A: D(A)\subset X\to X$ be a closed linear operator. A resolvent value of $A$ is a number $\lambda\in\mathbb{C}$ such that $\lambda I-A$ has a bounded inverse. The operator $A$ is called \emph{resolvent-positive} 
if there exists $\omega\in \mathbb{R}$ such that any number in the interval $(\omega, \infty)$ is a resolvent value of $A$ and $(\lambda I-A)^{-1}$ is a positive operator (i.e. maps $X_+$ to $X_+$) for any $\lambda\in (\omega, \infty)$. Denote $S(A)$ and $r(A)$ be the spectral bound and spectral radius of $A$, respectively. The following result is used to define $\mathcal{R}_0$.

\begin{lemma} [{\cite[Theorem 3.5]{thieme2009spectral}}] \label{lemma_thieme}
Let $A$ and $B$ be resolvent-positive operators on $X$. Suppose $A=B+C$, where $C: D(B)\to X$ is a positive linear operator. If $s(B)<0$, then $s(A)$ has the same sign as $r(-CB^{-1})-1$.
\end{lemma}

In our applications, we  take   $X:=\Pi_{j=1}^nC(\bar\Omega_j)$ and
\begin{equation}
\begin{split}
&X_{+}=\{\bm u=(u_1,\dots,u_n)\in X:u_j\ge0 \;\;\text{for}\;\;x\in\bar\Omega\}.
\end{split}
\end{equation}
For $\bm u=(u_1,\dots,u_n)\in X$, we write $\bm u\ge (\gg)\, \bm 0$ if $u_j(x)\ge (>) \, 0$ for all $x\in\bar\Omega_j$ and $j\in\Omega$ and $\bm u>\bm 0$ if $\bm u\ge \bm 0$ and $\bm u\neq \bm 0$. Define $A: X\rightarrow X$ be a linear operator given by 
\begin{equation}\label{An}
A(\bm u)=((A(\bm u))_1, \dots, (A(\bm u))_n)
\end{equation}
with
$$
(A(\bm u))_j=d_j\Delta u_j+d\sum_{k\neq j} \left(\f{L_{jk}\bar u_k}{|\Omega_j|}-  L_{kj}  u_j\right)+{a_j  u_j},
$$
where  $\bm u=(u_1, \dots, u_n)$, $a_j$ is H\"older continuous on $\bar\Omega_j$ for all $j\in\Omega$, $L=(L_{jk})$ satisfies (A2), and $d, d_1, \dots, d_n>0$. The domain $D(A)$ of $A$ is
$$
\ds D=D(A):=\prod_{j=1}^n \left\{u_j\in\bigcap_{p\ge 1}W^{2, p}(\Omega_j): \ \frac{\partial u_j}{\partial \nu}=0 \ \text{on} \ \partial\Omega_j \ \ \text{and}\ \ \Delta u_j\in C(\bar\Omega_j)\right\}.
$$

\begin{lemma}\label{lemma_positive}
The operator $A$ has a principal eigenvalue $\lambda_1:=s(A)$, which is simple and real and corresponds with a positive eigenvector; and $\lambda_1$ is the unique eigenvalue of $A$ corresponding with a positive eigenvector.  Moreover for each $\lambda>\lambda_1$, $\lambda I-A$ is invertible and $(\lambda I-A)^{-1}$ is strongly positive, i.e., $A\bm u\in int(X_+)$ for any $\bm u\in X_+$ with $\bm u\neq\bm 0$.  
\end{lemma}
\begin{proof}
 The operator $A$ is a bounded perturbation of $A'$, where $A': X\rightarrow X$ is given by 
$$
(A'(\bm u))_j=d_j\Delta u_j-d\sum_{k\neq j}   L_{kj}  u_j+{a_j  u_j},
$$
with $D(A)=D(A')=D$. 
 Moreover, $A'\le A$  on $X$, i.e. $A'\bm u\le A\bm u$ for any $\bm u\in X_{+}$.
It is well-known that $A'$ generates an analytic compact positive semigroup $T'(t)$ on $X$ ($T'(t)$ being positive means that $T'(t)\bm u\ge 0$ for any $t\ge 0$ and $\bm u\in X_+$). Therefore by \cite[Proposition 3.1.12, Corollary 6.1.11]{engel1999one}, the semigroup  $T(t)$ generated by $A$ is an analytic compact positive semigroup and $s(A)\ge s(A')$. Since $T(t)$ is positive,  
 $(\lambda I-A)^{-1}$ is positive if and only if $\lambda >s(A)$ \cite[Lemma 6.1.9]{engel1999one}. Since $T(t)$ is compact, $(\lambda I-A)^{-1}$ is compact for any $\lambda>s(A)$. 

Fixing $\lambda >s(A)$, we claim that $(\lambda I-A)^{-1}$ is strongly positive. Indeed, let $\bm v\in X_+$ with $\bm v\ne\bm 0$ and denote $\bm u=(\lambda I-A)^{-1}\bm v$. The positivity of $(\lambda I-A)^{-1}$  implies $\bm u\ge 0$. Since $(\lambda I-A) \bm u=\bm v$, we have 
\begin{equation}\label{eiggg}
-d_j\Delta u_j-d\sum_{k\neq j} \left(\f{L_{jk}\bar u_k}{|\Omega_j|}-  L_{kj}  u_j\right)+(\lambda-a_j)  u_j=v_j, \ \ x\in\Omega_j, j\in\Omega.
\end{equation}
Integrating it over $\Omega_j$, we obtain
$$
-d\sum_{k\in\Omega} L_{jk}\bar u_k+(\lambda-a_{jm})\bar u_j\ge \bar v_j, \ \ j\in\Omega,
$$
where $a_{jm}=\min\{a_j(x): \ x\in\bar\Omega_j\}$. 
Then, we have $\bar u_j>0$ for all $j\in\Omega$. Indeed, suppose to the contrary that $\bar u_{j_0}=0$ for some $j_0\in\Omega$. Then, we obtain $-d\sum_{k\neq j_0} L_{j_0k}\bar u_k\ge \bar v_{j_0}\ge 0$. This implies $\bar u_k=0$ for all $k\in\Omega$ with $k\neq j_0$ and $a_{j_0k}>0$. Since $L$ is irreducible, we must have $\bar u_k=0$ and $\bar v_k=0$ for all $k\in \Omega$. This contradicts $\bm v>\bm 0$ and proves $\bar u_j>0$ for all $j\in\Omega$. Then by \eqref{eiggg} and the elliptic maximum principle, we have $\bm u\gg \bm 0$. This proves that $(\lambda I-A)^{-1}$ is strongly positive.

We have shown that   $(\lambda I-A)^{-1}$ is compact and strongly positive for each $\lambda >s(A)$. 
By the Krein-Rutman theorem (\cite[Theorem 19.3]{deimling2010nonlinear}), $r((\lambda I-A)^{-1})$ is the principal eigenvalue of $(\lambda I-A)^{-1}$, which is simple and corresponds with a positive eigenvector and it is the unique eigenvalue corresponding with a positive eigenvector. Since $(\lambda I-A)^{-1}$ is compact, every spectral value of $A$ is an eigenvalue (\cite[Corollary 4.1.19]{engel1999one}). It is easy to check that $\lambda'$ is an eigenvalue of $(\lambda I-A)^{-1}$ if and only if $\lambda-1/\lambda'$ is an eigenvalue of $A$. Then the claimed results hold. 
\end{proof}

\begin{lemma}\label{lemma_B}
Let $A$ be defined in \eqref{An} and suppose that  $\bm a=(a_1, \dots, a_n)< \bm 0$. Then, $s(A)<0$. 
\end{lemma}
\begin{proof}
Let $\bm\varphi=(\varphi_1, \dots, \varphi_n)$ be a positive eigenvector of $A$ corresponding with the principal eigenvalue $\lambda_1=s(A)$. Then, we have 
\begin{equation}\label{eign}
\left\{
\begin{array}{lll}
\ds d_j\Delta \varphi_j+d\sum_{k\neq j} \left(\f{L_{jk}\bar \varphi_k}{|\Omega_j|}-  L_{kj}  \varphi_j\right)+{a_j  \varphi_j}=\lambda \varphi_j, \ \ &x\in\Omega_j, j\in\Omega,\\
\ds\frac{\partial \varphi_j}{\partial \nu}=0, \ \ &x\in\partial\Omega_j, j\in\Omega. 
\end{array}
\right.
\end{equation}
Integrating the first equation of \eqref{eign} over $\Omega_j$ and summing up for all $j\in\Omega$, we obtain 
$$
d\sum_{j, k\in\Omega} L_{jk}\bar\varphi_k + \sum_{j\in\Omega} \int_{\Omega_j} a_j \varphi_j =\lambda_1  \sum_{j\in\Omega} \bar\varphi_j. 
$$
By $\bm a<\bm 0$ and $\bm\varphi\gg \bm 0$, we have $\lambda_1<0$. 
\end{proof}

\subsection{Comparison principle}

\begin{lemma}\label{lemma_comp}
Suppose that $L=(L_{jk})$ satisfies (A2) and $d, d_1, \cdots, d_n>0$. Let $\bm u^i(\cdot, t)\in C([0, T]; X)$, $i=1, 2$, be two  solutions of 
 \begin{equation}\label{coml}
\begin{cases}
\ds\f{\partial  u_j}{\partial t}=d_{j} \Delta  u_j+d\sum_{k\neq j} \left(\f{ L_{jk}{\bar u_k}}{|\Omega_j|}- L_{kj}  u_j\right)+f_j(x, u_j),&x\in\Omega_j, t>0, j\in\Omega,\\
\ds\f{\partial u_j}{\partial \nu}=0, &x\in\partial\Omega_j, t>0, j\in\Omega,
\end{cases}
\end{equation}
where $f_j:\mathbb{R}^2\to\mathbb{R}$ is continuously differentiable for each $j\in\Omega$. 
If $\bm u^1(\cdot, 0)\le \bm u^2(\cdot, 0)$, then  $\bm u^1(\cdot, t)\le \bm u^2(\cdot, t)$ for all $t\in [0, T]$. If $\bm u^1(\cdot, 0)< \bm u^2(\cdot, 0)$, then $\bm u^1(\cdot, t)\ll \bm u^2(\cdot, t)$ for all $t\in (0, T]$. 
\end{lemma}
\begin{proof}
Suppose  $\bm u^1(\cdot, 0)\le \bm u^2(\cdot, 0)$.  Choose $c>0$ be sufficiently large and let $\bm {v}=e^{ct}(\bm u^1-\bm u^2)$. 
Then $\bm v$ satisfies 
 \begin{equation}\label{coml1}
\begin{cases}
\ds\f{\partial  v_j}{\partial t}=d_{j} \Delta  v_j+d\sum_{k\neq j} \left(\f{ L_{jk}{\bar v_k}}{|\Omega_j|}- L_{kj}  v_j\right)+v_j\left(\frac{\partial}{\partial y}f_j(x, \theta_j)-c\right),&x\in\Omega_j, t>0, j\in\Omega,\\
\ds\f{\partial v_j}{\partial \nu}=0, &x\in\partial\Omega_j, t>0, j\in\Omega,
\end{cases}
\end{equation}
where $\theta_j$ is between $u^1_j$ and $u^2_j$ for $j\in\Omega$. 
We want to show that $\bm v(\cdot, t)\le \bm 0$ for all $t\in [0, T]$. To see it, suppose to the contrary that this is not true. Then we have $\max_{j\in\Omega, x\in \bar\Omega_j, t\in [0, T]} v_j(x, t)>0$. Suppose that the maximum is attained at $j=j_0$ and $(x, t)=(x_0, t_0)\in \bar\Omega_{j_0}\times (0, T]$. If $x_0\in\Omega_{j_0}$, then $\partial_t v_{j_0}(x_0, t_0)\ge 0$, $\Delta v_{j_0}(x_0, t_0)\le 0$, and $\bar v_k(t_0)\le |\Omega_k| v_{j_0}(x_0, t_0)$ for all $k\in\Omega$. Therefore if $c>0$ is large, evaluating the first equation of \eqref{coml} at $j=j_0$ and $(x, t)=(x_0, t_0)$ will lead to a contradiction.  Otherwise, $x_0\in\partial\Omega_{j_0}$ and $v_{j_0}(x_0, t_0)>v_{j_0}(x, t)$ for all $(x, t)\in\Omega_{j_0}\times [0, t_0]$. We can choose $c>0$ large such that   $\partial_t v_{j_0}(x, t)-d_{j_0}\Delta v_{j_0}(x, t)<0$ near $(x_0, t_0)$. Then the Hopf boundary implies $\partial_n v_j(x_0, t_0)>0$, which is a contradiction. 

Now suppose  $\bm u^1(\cdot, 0)< \bm u^2(\cdot, 0)$. Let $c$ and $\bm v$ be defined as above. It suffices to show $\bm v(\cdot, t)\ll \bm 0$  for all $t\in (0, T]$. We have already shown that $\bm v(\cdot, t)\le \bm 0$ for all $t\in [0, T]$. Integrating the first equation of \eqref{coml1} on $\Omega_j$ for $j\in\Omega$, we obtain
 \begin{equation}\label{coml2}
\begin{cases}
\ds\f{d  \bar v_j}{d t}\le d\sum_{k\in\Omega} L_{jk}\bar v_k,& t>0, j\in\Omega,\\
\ds\bar {\bm v}(0)< \bm 0, &t>0.
\end{cases}
\end{equation}
Since $L$ is irreducible and essentially nonnegative, we have $\bar{\bm v}(t)\ll\bm 0$ for all $t>0$ (\cite{Smith1995}). Then by \eqref{coml1} and the parabolic comparison principle, we have $\bm v(\cdot, t)\ll \bm 0$  for all $t\in (0, T]$.
\end{proof}

{\large\noindent{\bf Declarations}}

\noindent{\bf Conflict of interest} The authors declare that they have no conflict of interest.

\bibliographystyle{plain}
\bibliography{epidem}

\end{document}